\documentclass[aps,pra,10pt,superscriptaddress,nofootinbib,twocolumn,floatfix]{revtex4-2}
\usepackage{graphicx}
\usepackage{amsmath,amsthm,amssymb,dsfont,enumitem}
\usepackage{multirow}
\usepackage{mdframed}
\usepackage{orcidlink}
\usepackage{appendix}
\usepackage{subcaption}
\usepackage{soul}
\usepackage{mathtools}
\newcommand\myeq{\stackrel{\mathclap{\normalfont\mbox{\tiny{NCHV}}}}{\leq}}
\newcommand\myeqq{\stackrel{\mathclap{\normalfont\mbox{\tiny{LHV}}}}{\leq}}
\usepackage{float}
\newcommand{\ket}[1]{\left| #1 \right\rangle}
\newcommand{\bra}[1]{\left\langle #1 \right|}

\newcommand{\braket}[2]{\langle #1|#2 \rangle}

\usepackage{hyperref}
\usepackage[capitalise,nameinlink]{cleveref}
\usepackage[authormarkuptext=name,commentmarkup=uwave]{changes}
\definechangesauthor[name={JRH},color=green!70!black]{jonte}

\newtheorem{theorem}{Theorem}
\definechangesauthor[name={EB},color=red!70!black]{Enrico}

\theoremstyle{definition} 
\newtheorem{definition}{Definition}[section]
\newtheorem{proposition}{Proposition}[section]
\newtheorem{assumpt}{Classical Assumption}
\begin{document}

\title{Warring Contextualities - Provably Classical vs Provably Nonclassical}
\author{Enrico Bozzetto}
\affiliation{DET, Politecnico di Torino, Corso Duca degli Abruzzi, 24, 10129 Torino, Italy}
\affiliation{Quantum Group, School of Computing, Newcastle University, 1 Science Square, Newcastle upon Tyne, NE4 5TG, UK}
\author{Jonte R. Hance\,\orcidlink{0000-0001-8587-7618}}
\email{jonte.hance@newcastle.ac.uk}
\affiliation{Quantum Group, School of Computing, Newcastle University, 1 Science Square, Newcastle upon Tyne, NE4 5TG, UK}

\begin{abstract}
 In the literature, there are two differing definitions of contextuality: Kochen and Specker's, and Spekkens' (or ``generalised''). However, researchers using one of these definitions rarely consider the other, meaning comparative analysis of these two notions is rare. In this paper, we advance the idea that Kochen-Specker contextuality provides a generalisation of the idea of system being fundamentally \emph{nonclassical}, while Spekkens' noncontextuality provides a generalisation of the idea of a system being \emph{classical}. This allows us to reconcile the two approaches, as different stages in a hierarchy of classicality/nonclassicality.
\end{abstract}

\maketitle

\section{Introduction}

Contextuality is one of the most peculiar feature of quantum mechanics. It is the principle that the value of an observable is determined by its surrounding context - how we choose to measure that observable, or what other observables we've measured on the system previously - rather than being intrinsic or measurement-independent, as we would expect physical properties to be.
It is an aspect of quantum systems that defies common sense, and for this reason it represents one of the most characteristic features of quantum systems. It also seems to be an aspect of quantum systems which is intimately tied to their practical usefulness, or potential quantum advantage~\cite{Raussendorf2013, Howard2014, shahandeh2021quantumcomputationaladvantage, Giordani2023Experimental, Flatt2026QAdvantage}. However, the question of how best to formally define contextuality has become somewhat contentious.

In the literature, there are two key definitions of (non)contextuality: Kochen-Specker (non)contextuality~\cite{KS1967,Budroni2022KSContextuality} (as has been built on by Larsson, Cabello~\cite{Cabello1996,CabelloSeveriniWinter2014,Cabello2015NecessaryAndSufficient,Cabello2018TIFS,Cabello2021ConvertingCtoNL,Cabello2012SI-CtoNonlocality}, Abramsky~\cite{Abramsky2011Sheaf}, Dzhafarov, etc), and Spekkens' ``generalised'' (non)contextuality~\cite{Spekkens2005,Spekkens2007Negativity} (as has been been built on by Leifer~\cite{Leifer2013}, Sainz, Selby, Schmidt~\cite{Schmid2021,Schmid2024structuretheorem,Schmid2025shadowssubsystemsof, Selby2023Incompat}, etc).
Given discussions around which of these notions of (non)contextuality should be used can often become somewhat polemical in papers from either of these two camps, researchers new to the area may struggle to understand the difference between the two, or their relative areas of applicability. 

In this paper, we advance the idea that Kochen-Specker contextuality provides a provable signature of a system being fundamentally \emph{nonclassical} in a way generalised from notions of Bell-nonlocality, while Spekkens' ``generalised'' noncontextuality provides a provable signature of a system being a generalised form of \emph{classical} (i.e., obeying some sort of tomographic completeness). 

This hopefully will allow us to reconcile the two approaches, as different stages in a hierarchy of classicality/nonclassicality (depending on which of these two directions one wishes to focus on). Such an understanding should help the two communities each move from viewing the alternative notion as a competitor, to instead viewing it as complementary, with each to be used within its own area of applicability.

This paper is laid out as follows. In Section~\ref{sec:reviewing}, we first review Kochen-Specker (Section~\ref{subsec: KS contextuality}) and Spekkens' ``generalised'' (non)contextuality (Section~\ref{subsec:Spekkens}), with specific focus on the Kochen-Specker noncontextual polytope (Section~\ref{subsubsec:KSnoncontextualpolytope}) and the simplex embeddability criterion for Spekkens' noncontextuality (Section~\ref{subsubsec:simplex}). In Section~\ref{sec:differences} then look at the differences between these two notions, and show that while Spekkens noncontextuality implies Kochen-Specker noncontextuality (and so, by contraposition, Kochen-Specker contextuality implies Spekkens contextuality), the inverse does not hold - Spekkens contextuality does not imply Kochen-Specker contextuality, and Kochen-Specker noncontextuality does not imply Spekkens noncontextuality (Section~\ref{subsec:KSCimpliesSpekkensC}). 

In Section~\ref{sec:provably}, we try to link this hierarchy of implication to ideas of classicality and quantumness. First, in Section~\ref{subsec: A classical system must be Spekkens noncontextual} we define sufficient criteria for a system to be ``classical'', and show that any system which meets these criteria must be Spekkens-noncontextual. Next, in Section~\ref{subsec:Bell}, we give Bell-nonlocality as a suitable sufficient condition for a system being nonclassical (and in Section~\ref{subsubsec:BellPolytope} introduce the Bell polytope as a geometric quantification of the border between Bell-locality and Bell-nonlocality). This allows us in Section~\ref{subsec:KSandBell} to use sheaf theory (Section~\ref{subsubsec:Sheaf}) to show Bell-nonlocality implies Kochen-Specker contextuality (Section~\ref{subsubsec:UsingSheafBellimpliesKS}), and to show that (in certain circumstances) Kochen-Specker contextuality implies Bell-nonlocality (Section~\ref{subsubsec:KSimpliesBellish}), allowing us to set Kochen-Specker contextuality as a sufficient condition for nonclassicality. In Section~\ref{subsec:SpekkensandBell} we then show that Spekkens contextuality does not imply Bell-nonlocality (Section~\ref{subsubsec:SpekkensnotimplyBell}) but that Bell-nonlocality implies Spekkens contextuality (Section~\ref{subsubsec:BellimpliesSpekkens}), setting Spekkens contextuality as a necessary but not sufficient condition for nonclassicality.

We then summarise all these relations between (non)classicality, Kochen-Specker (non)contextuality, Spekkens (non)contextuality, and Bell (non)locality in Section~\ref{subsec:RelationSummary}, before discussing the implications of these relations in Section~\ref{sec:discussion}.

\section{Reviewing Contextualities}\label{sec:reviewing}

\subsection{Kochen-Specker Contextuality}
\label{subsec: KS contextuality}
The concept of contextuality in quantum mechanics was first formalised by Bell in 1966 \cite{Bell1966} and then by Kochen and Specker in 1967 \cite{KS1967}. Their seminal work, now known as the (Bell-)Kochen-Specker theorem, demonstrates that quantum mechanics is incompatible with non-contextual hidden variable models that assume outcome determinism.

Consider indeed a $d$-dimensional Hilbert space $\mathcal{H}$ with a set of $d$ rank-1 projectors $\{P_i\}_{i=1}^d$ associated with an orthonormal basis. Assume these projectors satisfy the conditions of orthogonality and completeness:
\begin{equation}
    P_i P_j = 0 \quad \text{for} \,\, i \neq j, \quad \text{and} \quad \sum_{i=1}^d P_i = \mathds{1}.
\end{equation}
Given the above relations, projectors in the set obviously commute ($\forall i, j, \;[P_i,P_j]=0$).
These projectors are typically interpreted as a set of mutually exclusive logical propositions, such that we would expect them to be able to be assigned a Boolean truth value - i.e., 0 or 1. Two different propositions $Q_i,Q_j$ are exclusive, i.e. they cannot be simultaneously "true" for $i\neq j.$. Moreover $Q_1,...,Q_d$ cannot be all simultaneously "false"; one of them must be true.

A ``context'' is defined as a set of $d$ mutually commuting observables that can be measured simultaneously. Notably, for $d\geq3$, a given projector $P_\alpha$ can be in multiple different contexts - i.e., can be orthogonal to two other projectors $P_\beta$ and $P_\gamma$ which are not themselves orthogonal, and so do not commute ($[P_\alpha,P_\beta]=[P_\alpha,P_\gamma]=0$, but $[P_\beta,P_\gamma]\neq0$). We can treat the ``measurement'' of a context as simultaneously assigning a value to all projectors in that context - given their orthogonality, an assignment of value 1 to one of the projectors, and an assignment of value 0 to all the other projectors in that context.

The Kochen-Specker theorem explores this idea of projectors for $d\geq3$ being able to be in multiple different contexts. It was originally formulated in the simplest scenario where it can be formulated, i.e. $d=3$~\cite{KS1967}. The authors provided a physical interpretation of certain rank-1 projectors in $d=3$ as spin operators for a spin-1 particle. In this way they replaced complex vectors belonging to Hilbert space with vectors $\vec{v}\in \mathbb{R}^3$. The theorem states that it is impossible to assign a definite value (0 or 1) to all the projectors in this set in a way that is independent of the context in which they are measured.

\begin{theorem}[Kochen and Specker, 1967]
There exists a finite set $S \subset \mathbb{R}^3$ such that no value-assignment function $f: S \rightarrow \{0,1\}$ can satisfy
\begin{equation}
    f(\vec{u}) + f(\vec{v}) + f(\vec{w}) = 1
\end{equation}
for all triplets $(\vec{u}, \vec{v}, \vec{w})$ of mutually orthogonal vectors in $S$.
\end{theorem}
There are two ways to explain this result: the value assigned to an observable must either be non-existent prior to measurement, or must change depending on which other compatible observables are being measured alongside it.
In quantum mechanics, projectors onto the state of a system are used to determine the truth values of propositions regarding its physical properties. Therefore this theorem demonstrates that any assignment of these truth values is inherently dependent on the context in which the projectors are applied. If we assume that the properties of a physical system possess a value also before a measurement, then these values must  depend on the context in which they are measured. Properties are contextual.

While the original proof used 117 vectors~\cite{KS1967}, subsequent work has reduced this to much smaller sets, though the requirement for the dimension of the Hilbert space remains $d \geq 3$. This is because the theorem requires that there exist at least two orthogonal projectors for every projector, and so that a projector can belong to at least two different contexts, which is not satisfiable in a two dimensional Hilbert space.

The simplest example in which we can witness contextuality is the KCBS scenario \cite{Klyachko2008Simple}. This is defined by five measurements in a three dimensional Hilbert space $\{A_0,..,A_4\}$, all with outcomes $a_i\in \{-1,1\}$, where each pair $A_i$ and $A_{i+1}$ (modulo 5) are compatible. For a noncontextual hidden variable model (NCHV), where the results of any measurement reflect measurement-independent properties of the system (represented by a hidden variable $\lambda)$, the following inequality must be valid:
\begin{equation}
\label{eq: KCBS}
    \langle A_0A_1\rangle+\langle A_1A_2\rangle+\langle A_2A_3\rangle+\langle A_3A_4\rangle+\langle A_4A_0\rangle \geq -3
\end{equation}
where $\langle A_iA_j\rangle=\sum_{a_i,a_j}a_ia_jp(a_i,a_j)$.
However if one considers the state $\ket{\psi}=\ket{0}$ and measurement settings 
\begin{equation}
    \begin{split}
        &A_j=2\ket{v_j}\bra{v_j}-\mathds{1},\text{ where}\\
        &\ket{v_j}=\cos(\theta)\ket{0}+ \sin(\theta) \cos(\frac{4\pi j}{5})\ket{1}\\
        &\;\;\;\;\;\;\;\;\;\;\;+\sin(\theta) \sin(\frac{4\pi j}{5})\ket{2},\\
        &\text{ for }\cos^2(\theta)=\frac{\cos(\pi/5)}{1+\cos(\pi/5)},
    \end{split}
\end{equation}
the left-hand side of the inequality in Eq.~\eqref{eq: KCBS} can equal $5-4\sqrt{5}\thickapprox -3.94$. The violation of the KCBS equation is the simplest proof of the impossibility of having a NCHV model for quantum mechanics.

While the notion of a context as a set of compatible observables which we used above works well for theory, in practice it is harder to apply, as it is difficult experimentally to prove that two observables are compatible. Therefore, there are two main definitions of Kochen-Specker contexts in the literature: the OP (Observable Perspective) definition, and the EP (Effect Perspective) definition.
In the first case, the basic components remain observables, while a context becomes a set of observables that satisfy outcome repeatability, as well as statistical non-disturbance conditions~\cite{Heinosaari_2010} and their generalisation to arbitrary sequences.\\
In the EP definition instead a context is identified operationally, as consisting simply of a single measurement. The basic components are then the eigenstates of this measurement. This definition however leaves us with the problem of identifying the same effect in different measurements. In this view Spekkens's definition of contextuality (see Section~\ref{subsec:Spekkens}) plays an important role, given it identifies effects as being the same based on their observed statistics.

The OP definition can be expanded to treat also the cases of real, non-perfectly compatible measurements. 
An explanation of this approach is present in Appendix \ref{appendix: KS contextuality for non-perfectly commuting observables}.

\subsubsection{Kochen-Specker Noncontextual Polytope}\label{subsubsec:KSnoncontextualpolytope}
It is possible to represent geometrically the set of all possible behaviours that can be achieved within a
Kochen-Specker noncontextual model. This set forms a polytope - the Kochen-Specker noncontextual polytope ($\mathcal{P}_{KS}$)~\cite{Budroni2022KSContextuality}. 
One of the fundamental properties of a polytope is that it can be defined both using its vertices and its facets. We can define the noncontextual polytope as the convex hull of all deterministic noncontextual assignments of events, or as a finite collection of facet inequalities, i.e. noncontextuality inequalities.\\
Noncontextuality inequalities provide bounds obeyed by noncontextual hidden-variable models, in analogy with Bell inequalities~\cite{Bell1966} that provide bounds for local hidden-variable models. Ref.~\cite{Badziag2009} proved that each Kochen-Specker set (a set of vectors which does satisfy Kochen-Specker theorem) can be converted into an inequality with the following structure
\begin{equation}
    \sum_i\langle A_i \rangle-\frac{1}{2}\sum_{(i,j)}\langle A_iA_j\rangle\leq B_{NC}
\end{equation}
where $\{A_i\} $ is a set of observables and $B_{NC}$ is the maximum value allowed by the noncontextual hidden-variable model. In particular $B_{NC}=n(d-2)-2$ where $n$ is the number of observables and $d$ the number of different contexts~\cite{Badziag2009}.

Given a specific set of $n$ measurement settings and $n$ outcomes, we have $n$ inequalities that must be respected by the system to be KS noncontextual. Each of these inequalities form a face of the noncontextual polytope. This polytope will be useful later for proving relations between Kochen-Specker (non)contextuality and other properties of systems.

\subsection{Spekkens' Generalised Contextuality}
\label{subsec:Spekkens}
A nominally-broader notion of contextuality was introduced by Spekkens in 2005 \cite{Spekkens2005}. This framework generalises the Kochen-Specker definition by shifting from a logic of projectors to an operational representation of experiments. In this model, the primitive elements are experimental procedures: preparations ($P$), transformations ($T$), and measurements ($M$).

Central to this definition is the concept of Operational Equivalence. Two preparation procedures $P$ and $P'$ are considered (operationally) equivalent if they yield identical statistics for all possible measurements:
\begin{equation}
    P \sim P'\leftrightarrow p(k|P,M) = p(k|P',M) \quad \forall\  M, k.
\end{equation}
Similarly, two measurements $M$ and $M'$ are equivalent if they yield the same statistics for all possible preparations:
\begin{equation}
    M\sim M'\leftrightarrow p(k|P,M)=p(k|P,M') \quad \forall\ P,k
\end{equation}
Finally, two transformations are equivalent if they yield the same statistics for all possible preparations and measurements:
\begin{equation}
    T\sim T'\leftrightarrow p(k|P,T,M)=p(k|P,T',M) \quad \forall\  P,M,k
\end{equation}

An ontological model of an operational theory is defined as a map from the operational procedures of a scenario (preparations, transformations, measurements) to a deeper, underlying reality, whose attributes are defined regardless of what anyone knows about them.
In particular an ontological model is defined by three primary elements: a space of ontic states $\lambda \in\Lambda$ which contains all possible states the system can be in, a function that assigns each quantum state $\rho$ to a probability measure $\mu_\rho$ over this ontic state space, and a function that assigns each quantum observable to a Markov kernel for each ontic state in this state space.\footnote{A Markov kernel from $(X,\Sigma_X)$ to $(Y,\Sigma_Y)$ is a map that assigns to every $x \in X$ a probability measure $K(x,\cdot)$ on $(Y,\Sigma_Y)$ such that for every measurable set $A \in \Sigma_Y$, the function $x \mapsto K(x,A)$ is $\Sigma_X$-measurable.} The necessary condition for the validity of the ontological model is that the model must reproduce the same statistical predictions as the operational theory we are considering \cite{Tezzin2025Luders}.
The probability of an outcome is therefore given by the law of total probability:
\begin{equation}
    p(k|P,T,M) = \int_\Lambda d\lambda d\lambda'\, \mu_P(\lambda) \Gamma_T(\lambda',\lambda) \xi_{M,k}(\lambda')
    \label{equation ontological procedures representation}
\end{equation}
where $\mu_P(\lambda):\Lambda\rightarrow[0,1]$ is the probability distribution associated with the preparation $P$, $\Gamma_T(\lambda',\lambda):\Lambda\times \Lambda\rightarrow[0,1]$ is the transformation matrix which represents the probability of moving from the ontic state $\lambda$ to the state $\lambda'$ and $\xi_{M,k}(\lambda'):\Lambda\rightarrow[0,1]$ is the probability of obtaining the outcome $k$ performing the measurement $M$, given the ontic state $\lambda'$.

In the ontological model then, the condition of equivalence between two preparations can be expressed as follows
\begin{equation}
\label{eq: prep equivalence}
\begin{aligned}
    P\sim P' \leftrightarrow    & \int_\Lambda \xi_{M,k}(\lambda)\,\mu_{P}(\lambda)d\lambda
    =\\&
    \int_\Lambda \xi_{M,k}(\lambda)\,\mu_{P'}(\lambda)d\lambda
    \quad \forall\,M,k
\end{aligned}
\end{equation}
The condition of equivalence between two measurement instead can be expressed as
\begin{equation}
\label{eq: meas equivalence}
\begin{aligned}
    M\sim M' \leftrightarrow     &\int_\Lambda \xi_{M,k}(\lambda)\,\mu_P(\lambda)d\lambda
    =\\&
    \int_\Lambda \xi_{M',k}(\lambda)\,\mu_P(\lambda)d\lambda
    \quad \forall\,P, k
    \end{aligned}
\end{equation}
The condition of equivalence between two transformations at the end, can be expressed as
\begin{equation}
\label{eq: transf equivalence}
\begin{aligned}
    &T\sim T' \leftrightarrow\\
    &\int_\Lambda d\lambda d\lambda'\, \mu_P(\lambda) \Gamma_{T}(\lambda',\lambda) \xi_{M,k}(\lambda')=\\&\int_\Lambda d\lambda d\lambda'\, \mu_P(\lambda) \Gamma_{T'}(\lambda',\lambda)\xi_{M,k}(\lambda')\quad \forall P,M,k
\end{aligned}
\end{equation}

It is possible to define an equivalence class of procedures as a set of equivalent procedures and a context as the set of features that are not specified by specifying the equivalence class. Hence an ontological model is defined noncontextual if every experimental procedure depends only on its equivalence class, and not on its context. It is possible to rephrase the previous statement to the following equivalent one: a model is defined as noncontextual if operationally equivalent procedures have identical ontological representations. This means that
\begin{equation}
\begin{split}
    &P\sim P'\Rightarrow\mu_P=\mu_{P'},\\& T\sim T'\Rightarrow \Gamma_T=\Gamma_{T'},\\& M\sim M'\Rightarrow\xi_{M,k}=\xi_{M',k}
    \end{split}
\end{equation}
In particular a system is defined as (Spekkens) preparation contextual if the first condition does not hold, (Spekkens) transformation contextual if the second condition does not hold and (Spekkens) measurement contextual if the third condition does not hold. Moreover a system is defined as (Spekkens) universally contextual if all the three previous conditions do not hold.

Quantum mechanics has been proven to be Spekkens contextual since it does not admit a noncontextual ontological model~\cite{Spekkens2005}.

\subsubsection{Simplex embeddability}\label{subsubsec:simplex}

Simplex embeddability was introduced in Ref.~\cite{Schmid2021} as a necessary and sufficient geometric condition for defining a system as Spekkens noncontextual. To describe it we must firstly define the concept of generalised probabilistic theory (GPT). A GPT is a theory that we can obtain from an operational theory by discarding information about those experimental procedures which can be varied without affecting the operational statistics. A GPT associates to a system a convex set of states, $\Omega$, which live in a inner product space $(V,\langle\_,\_\rangle)$ one dimension higher than the affine space of $\Omega$. A GPT also associates to a system a set of effect vectors $\mathcal{E}$ such that the probability of obtaining an effect $e \in \mathcal{E}$ given the state $s\in \Omega$ is given by the scalar product $\langle e,s\rangle$. If one defines the dual of $\Omega$, denoted $\Omega^*$, as the set of vectors in $V$ whose inner product with all state vectors in $\Omega$ is between 0 and 1, then we require that $\mathcal{E}\subseteq\Omega^*$. The GPT must also satisfy the assumption of \textbf{Tomographic Completeness}: the GPT states and GPT effects must be uniquely identifiable by the probabilities that they produce. This means that if two states give the same probability as each other for every possible effect, then the two states must be the same; and if two effects give the same probability if applied to any state, they must be the same effect. Mathematically
\begin{equation}
    (\langle\textbf{e},\textbf{s}_1\rangle=\langle\textbf{e},\textbf{s}_2\rangle\ \forall \textbf{e} \in \mathcal{E})\leftrightarrow (\textbf{s}_1=\textbf{s}_2)
\end{equation}
for the states, and
\begin{equation}
\left(\langle\textbf{e}_1,\textbf{s}\rangle=\langle\textbf{e}_2,\textbf{s}\rangle\ \forall \textbf{s} \in \Omega\right)\leftrightarrow (\textbf{e}_1=\textbf{e}_2)
\end{equation}
for the effects. 

A GPT $G$ is therefore defined by the quadruple
\begin{equation}
G:=(V,\langle\_,\_\rangle,\Omega,\mathcal{E})    
\end{equation}
satisfying these constraints.

A GPT can be associated to an operational theory $T$ specifying two maps: 
\begin{equation}
    \textbf{s}\_:\{P\}\rightarrow\Omega\quad  \quad \textbf{e}\_:\{E\}\rightarrow\mathcal{E}
\end{equation}
where $\{P\}$ represents the set of possible preparations described by the operational theory and $\{E\}$ the set of possible effects described by the operational theory. These maps satisfy the constraint
\begin{equation}\label{eq.GPTProbRule}
    p([k|M],P)=\langle\textbf{e}_{[k|M]},\textbf{s}_P\rangle \ \forall P\in\{P\} ,E\in\{E\}
\end{equation}
where the left hand side is the probability of obtaining a result $k$ when performing measurement $M$ on a system prepared with preparation $P$ in the operational theory.
The Tomographic Completeness assumption for the GPT ensures
\begin{equation}
    P\simeq P'\iff \textbf{s}_P=\textbf{s}_{P'}
\end{equation}
and 
\begin{equation}
    [k|M]\simeq [k'|M']\iff \textbf{e}_{[k|M]}=\textbf{e}_{[k'|M']}
\end{equation}
Previously in this section, we defined the concept of an ontological model of an operational theory. This concept can be extended to GPTs.
An ontological model of a GPT associates to each GPT state vector $\textbf{s}\in \Omega$ a normalised probability distribution over the set of all ontic states $\Lambda$, denoted $\tilde{\mu}_\textbf{s} \in \mathcal{D}[\Lambda]$, and to each GPT effect vector $\textbf{e}\in \mathcal{E}$ a response function on $\Lambda$, denoted $\tilde{\xi}_\textbf{e} \in \mathcal{F}[\Lambda]$. Equivalently, the ontological model defines a pair of maps
\begin{equation}
    \tilde{\mu}\_:\Omega \rightarrow\mathcal{D}[\Lambda] \quad \text{and}\quad \tilde{\xi}\_:\mathcal{E}\rightarrow\mathcal{F}[\Lambda]
\end{equation}
where $\mathcal{D}[\Lambda]$ is the set of possible probability distribution associated to a preparation $P$ of the operational theory and $\mathcal{F}[\Lambda]$ is the set of response functions associated to a measurement $M$ with output $k$ of the operational theory. The two maps must satisfy the condition that the ontological model reproduces the probability rule of the GPT, (Eq.~\ref{eq.GPTProbRule}), so
\begin{equation}
    \langle\textbf{e},\textbf{s}\rangle=\sum_{\lambda\in \Lambda}\tilde{\xi}_\textbf{e}(\lambda)\tilde{\mu}_\textbf{s}(\lambda)
\end{equation}

From this, Ref.~\cite{Schmid2021} showed that:
\begin{proposition}
\label{prop: non contextual GPT}
    \textit{There exists a generalised-noncontextual ontological model of an operational theory $T$ describing prepare-measure experiments on a system iff there exists an ontological model of the GPT $G$ that $T$ defines.}
\end{proposition}
This sufficient and necessary condition can be transformed into a geometrical condition, respected by every noncontextual operational theory. This can be done using the concept of simplex\footnote{A simplex is the generalisation of the triangle (from 2-dimensional space) and tetrahedron (from 3-dimensional space) to any given $n$-dimensional space, representing the simplest possible polytope with $n+1$ vertices.} embeddability.
 
\begin{definition}
    \textit{A GPT describing a prepare-measurement experiment, $G=(V,\langle\_,\_\rangle_V,\Omega,\mathcal{E})$ is simplex-embeddable iff there exists:
    \begin{enumerate}[label=\roman*).]
        \item an inner product space $(W,\langle\_,\_\rangle_W)$ of some dimension d that contains a (d-1)-dimensional (hence d-vertex) simplex $\Delta_d$ (whose affine span does not contain the origin) and its dual hypercube $\Delta_d^*$; and,
        \item a pair of linear maps $\iota,\kappa:V\rightarrow W$ satisfying
    \begin{equation}
        \iota(\Omega)\subseteq \Delta_d,
    \end{equation}
    \begin{equation}
        \kappa (\mathcal{E})\subseteq \Delta^*_d,
    \end{equation}
    \begin{equation}
\langle\textbf{e},\textbf{s}\rangle_V=\langle\kappa(\textbf{e}),\iota(\textbf{s})\rangle_W\quad\forall \textbf{e}\in \mathcal{E},\textbf{s}\in \Omega.
    \end{equation}
    \end{enumerate}
    }
\end{definition}

This leads to Ref.~\cite{Schmid2021}'s key result:
\begin{theorem}
\label{theorem: simplex embeddability}
    \textit{A GPT describing a prepare-measure experiment admits an ontological model over an ontic space $\Lambda$ of finite cardinality iff that GPT is simplex-embeddable.}
\end{theorem}

Theorem~\ref{theorem: simplex embeddability} and Prop.~\ref{prop: non contextual GPT} together provide a necessary and sufficient geometric condition to establish whether an operational theory (described by a set of preparations and measurements which respect the Tomographic Completeness assumption) is Spekkens noncontextual; which is embedded in the following theorem
\begin{theorem}
For a prepare-measure experiment, the operational theory describing it admits of a generalised-noncontextual ontological model on an ontic state space of finite cardinality if and only if the GPT describing it is simplex-embeddable.
\end{theorem}

To test whether an operational model of a scenario obeys this criterion, we first determine the set of GPTs that are compatible with the data obtained from prepare-measure experiments for that scenario. We then just need to test whether these GPTs are simplex-embeddable. If all of them are simplex-embeddable, the operational theory is Spekkens noncontextual; if at least one of them is not simplex-embeddable, the operational theory is Spekkens contextual.

Notably, no lower limit had originally been set for the possible size of the simplex. This issues was solved in Ref.~\cite{Schmid2024structuretheorem} where the authors proved the lower bound to be exactly the dimension of the GPT itself.

The only assumption this simplex embeddability condition requires (above and beyond those required generally for ontological models \cite{pusey2012reality,Oldofredi2020,Hermens2021,Hance2022,Hance2022Ensemble,Carcassi2024,Tezzin2025Luders}) is Tomographic Completeness. However in Appendix~\ref{appendix: relaxing Tomography} it is shown that the assumption of Tomographic Completeness is in fact not necessary to demonstrate Spekkens noncontextuality via simplex embeddability.

\section{Differences between Contextualities}\label{sec:differences}

The two notions of contextuality described in the previous section differ in several important respects.
The Kochen–Specker notion is formulated within quantum mechanics, and concerns the impossibility of assigning noncontextual, deterministic values to projective measurements in Hilbert spaces of dimension greater than two. In its original version it presupposes outcome determinism, and applies specifically to sharp (i.e., strong projective) measurements. However Ref. ~\cite{Budroni2022KSContextuality} shows that a generalisation to nondeterministic outcomes is possible given that nondeterministic response functions can be always transformed into deterministic functions of a new hidden variable. Moreover a generalisation to non-perfectly compatible measurements is possible, as shown in Appendix~\ref{appendix: KS contextuality for non-perfectly commuting observables}.

In contrast, Spekkens’ notion is formulated within the general framework of operational theories and ontological models. It does not assume determinism a priori and it applies to arbitrary operational procedures, including preparations, transformations, and unsharp measurements.

For the case of unsharp measurements or POVMs (positive operator-valued measures), we can always evaluate whether a scenario is Spekkens contextual since no assumptions on ideal projective measurements have been made. Kochen-Specker in its original form, instead, cannot be applied in this case. However in \cite{hance2025externalquantumfluctuationsselect} it is shown that, using Naimark dilations, we can extend the definition of Kochen-Specker contextuality to include unsharp measurements or POVMs. In Appendix~\ref{sec: appendix POVMs for KS} we provide a brief review of this approach.

Another difference is that while Kochen–Specker contextuality cannot arise in two-dimensional Hilbert spaces (single qubits can be described by a Kochen-Specker NCHV model), Spekkens’ generalised notion can show contextual features even for qubit systems. Ref.~\cite{Spekkens2005} shows any maximally-mixed qubit is contextual in this sense.

However Johansson and Larsson proposed a model, Quantum Simulation Logic (QSL) \cite{Johansson2019QSL}, which mimics the behaviour of a (pure or maximally mixed) qubit, despite being entirely classical (using two bits). In Appendix \ref{appendix:qubit}, we present the QSL model for the qubit, and show it is noncontextual both in the Kochen-Specker and Spekkens sense. 

By the Kochen-Specker definition, both the models (qubit in a two dimensional Hilbert space, and QSL model of a qubit) are noncontextual - which aligns with Bell's earlier result that a qubit could always be represented by a noncontextual hidden variable model~\cite{Bell1966}. However, by the Spekkens definition, while a qubit can be contextual, its quantum simulation logic representation is instead always noncontextual. This reinforces the idea that a system being Spekkens-noncontextual can provide a signature of that system being definitely classical.

\subsection{Kochen-Specker Contextuality implies Spekkens Contextuality}\label{subsec:KSCimpliesSpekkensC}
In this section we prove that the notion of Spekkens contextuality is a generalisation of the Kochen-Specker notion of contextuality, insofar as a system being Kochen-Specker contextual implies it is also Spekkens contextual. We show this by geometrically comparing the Spekkens noncontextuality simplex with the Kochen-Specker noncontextuality polytope, and showing, for any scenario consisting of sharp measurements, the polytope always contains the simplex.

\begin{theorem}
\textit{If a generalised probabilistic theory $G=(V,\langle\_,\_\rangle_V,\Omega,\mathcal{E})$ is simplex-embeddable, then for any scenario consisting of sharp measurements, the generated behaviour $\mathcal{P}$ belongs to the Kochen-Specker noncontextual polytope $\mathcal{P}_{KS}$.}
\end{theorem}

\begin{proof}
 Consider a set of sharp measurements in $G$, grouped into contexts $C$. Each measurement outcome corresponds to an effect $\textbf{e} \in \mathcal{E}$. The behaviour $\mathcal{P}$ is given by the probability rule:
\begin{equation}
p(e|s) = \langle\textbf{e},\textbf{s}\rangle_V
\end{equation}
 where $\textbf{s} \in \Omega$ is the prepared state. By hypothesis, $G$ is simplex-embeddable, meaning there exist linear maps $\iota$ and $\kappa$ such that $\iota(\textbf{s}) \in \Delta_d$ and $\kappa(\textbf{e}) \in \Delta^*_d$.

Following the geometric properties of the simplex, we can decompose the mapped state into a unique convex combination of the $d$ vertices of the simplex, $\iota(\textbf{s}) = \sum_{\lambda=1}^d q_\lambda\textbf{v}_\lambda$, where each vertex $\textbf{v}_\lambda$ represents an ontic state. The probability rule becomes:
\begin{equation}
p(e|s) = \sum_{\lambda=1}^d q_\lambda\langle\kappa(\textbf{e}),\textbf{v}_\lambda\rangle_W
\end{equation}

In the Kochen-Specker paradigm, we restrict our attention to sharp measurements. 

For any sharp effect $\textbf{e}$ and any ontic state $\textbf{v}_\lambda$, the inner product evaluates strictly to Boolean values:
\begin{equation}
\langle\kappa(\textbf{e}),\textbf{v}_\lambda\rangle_W \in \{0, 1\}
 \end{equation}
This comes from the result obtained in Ref.~\cite{Spekkens2005} in which the author proved outcome determinism for sharp measurements, i.e., showed that preparation noncontextuality forces any projective measurement to be represented by $\{0,1\}$-valued response functions.

Let us define this deterministic assignment as $v(e, \lambda) \equiv \langle \kappa(\textbf{e}), \textbf{v}_\lambda \rangle_W$. Crucially, this value depends only on the specific effect $\textbf{e}$ and the ontic state $\lambda$. It is completely independent of the context $C$ (the set of other co-measurable effects) in which $\textbf{e}$ is evaluated.

Because the probabilities within any valid measurement context $C$ must sum to $1$, the deterministic assignments must also satisfy $\sum_{\textbf{e} \in C} v(e, \lambda) = 1$ for every ontic state $\lambda$. Thus, each ontic state $\lambda$ defines a global, context-independent, deterministic value assignment across all measurements.

Substituting $v(e, \lambda)$ back into our probability expression yields:
\begin{equation}
p(e|s) = \sum_{\lambda=1}^d q_\lambda v(e, \lambda)
\end{equation}

This equation demonstrates that the overall behaviour $\mathcal{P}$ is geometrically constrained to be a convex mixture of global, deterministic, noncontextual assignments. By definition, the set of all such convex combinations forms the Kochen-Specker noncontextual polytope. Therefore, $\mathcal{P} \in \mathcal{P}_{KS}$.
\end{proof}

Note obviously $\mathcal{P} \in \mathcal{P}_{KS}$ does not imply $\mathcal{P} \ni \mathcal{P}_{KS}$, so this result does not mean that Kochen-Specker noncontextuality implies Spekkens noncontextuality.

\section{``Provably Classical'' vs ``Provably Nonclassical''}\label{sec:provably}

In this section we formalise the intuition that Spekkens-noncontextuality provides some signature of a system being classical, while Kochen-Specker contextuality provides a signature that a system is nonclassical. To being able to link the concept of (non)contextuality to the concept of (non)classicality we need to define first of all what a classical and a nonclassical system are. Finding a strict dividing line between classical and nonclassical is one of the key interpretational problems in quantum mechanics, as no clear boundary (or ``Heisenberg cut'') has been universally agreed upon. However, there are several distinct features typically associated with quantum systems versus classical systems.

The first major distinction is determinism. Classical mechanics is fundamentally deterministic: if one knows the initial position and momentum of an object and the forces acting upon it, it is always possible to predict its future behaviour. In classical mechanics, the role of probability simply reflects a lack of knowledge about the exact microstate of the system. This means that classical probabilistic-ness is inherently epistemic and not ontological. In quantum mechanics, however, this is not true. Even at a fundamental level, given a complete description of a quantum state, the theory is generally at least instrumentally only able to yield probabilities for finding specific results of a measurement, as dictated by the Born rule. 

A second, related difference concerns the nature of the state space. If we define a state space of different values of properties possessed by a system, classical systems always occupy a definite state (a specific point in phase space). Conversely, unless we allow the state in state space to be context-dependent, quantum systems can exist in a simultaneous superposition of multiple states within a Hilbert space.

This divergence in state space leads directly to a third distinction regarding measurement, preparation limits, and the commutativity of observables. In classical mechanics, all observable properties commute, meaning infinitely precise preparation and measurement procedures can, in principle, be performed simultaneously. In quantum mechanics, this is fundamentally prohibited for conjugate variables (such as position $\hat{x}$ and momentum $\hat{p}$) which do not commute ($[\hat{x}, \hat{p}] = i\hbar$). This gives rise to the Heisenberg uncertainty principle, which bounds the precision with which these properties can be simultaneously defined:
\begin{equation}
    \Delta x \Delta p \geq \frac{\hbar}{2}
\end{equation}
It is vital to distinguish this fundamental quantum uncertainty from classical ignorance. Even in areas of classical physics where preparing a system in a given state with infinite precision is practically impossible (e.g., statistical mechanics and thermodynamics), one never questions that the system is in a definite state with all properties simultaneously defined. Our uncertainty about their values is purely epistemic. In quantum mechanics, however, as we proved above, the act of assuming that all properties possess simultaneously well-defined, measurement-independent values inherently leads to logical contradictions. 

Finally, classical and quantum systems differ fundamentally in how they compose. Classical composite systems can always be completely described by the individual states of their constituent parts. Quantum mechanics, however, allows for entangled states, where the state of the overall system is completely well-defined, yet the states of the individual subsystems are not. This leads to non-local correlations that violate classical bounds, further demonstrating that quantum systems cannot be underpinned by classical, locally real hidden variables.

Therefore, to formalise our intuition above, we first need to demonstrate that a system with the properties we associate with classicality must always be Spekkens noncontextual.
It is important to note the following: Spekkens' definition of contextuality requires the presence of an ontological model, which describes the operational theory we are dealing with. Given it is operational, talking about an ontological model of a classical theory also makes sense, because an ontological model goes beyond the difference between a classical and a quantum system. An ontological model corresponds to the idea of the properties of the system we are considering being defined regardless our knowledge about them. This idea can clearly be applied to classical physics. An example can be given by thermodynamics where we are able to measure the properties of a system such as the temperature, while in reality all these properties are defined by all the particles that compose the system and to which we don't have direct access. A detailed analysis of this can be found in Ref.~\cite{Bittner2018FormalOntology}.

\subsection{A classical system must be Spekkens noncontextual}
\label{subsec: A classical system must be Spekkens noncontextual}
We can model the state of a classical system as a point in phase space \(\Phi\), whose points are \(\gamma = (\vec r,\vec p) \in \Phi\). We assume the existence of an ontological model such that to each of these points in phase space \(\gamma\in\Phi\) there is associated unambiguously an ontic state \(\lambda\) in the ontic state space \(\Lambda\). Preparations are represented by probability distributions \(\mu_P\) on \((\Lambda,\Sigma)\), where \(\Sigma\) is the \(\sigma\)-algebra of measurable subsets of \(\Lambda\)\footnote{A $\sigma$-algebra on a set $X$ is a nonempty collection $\Sigma$ of subsets of $X$ closed under complement, countable unions and countable intersections.}. Measurements are represented by response functions and transformations by transition kernels.

We introduce the assumptions that characterise what we mean here by a classical system.

\begin{assumpt}\label{assumptclass1}
    \textit{The sets of preparations, transformations and measurements are closed under convex combination and coarse-graining.}
\end{assumpt}

\begin{assumpt}\label{assumptclass2}
    \textit{The ontic state space \((\Lambda,\Sigma)\) is a phase space (or more generally a standard measurable space) with its
    \(\sigma\)-algebra.}
\end{assumpt}

\begin{assumpt}\label{assumptclass3}
    \textit{For every measurable set \(\omega \in \Sigma\) there exists a
    two-outcome (i.e., dichotomous) measurement \(M_\omega\) with outcomes \(k=0,1\) such that
    the response functions in the ontological model are
    \begin{equation}\label{eq:responseindicatorclassical}
        \xi_{M_\omega,1}(\lambda) = \mathbb I_\omega(\lambda),
        \qquad
        \xi_{M_\omega,0}(\lambda) = 1 - \mathbb I_\omega(\lambda),
    \end{equation}
    where \(\mathbb I_\omega\) is the indicator function of \(\omega\),
    \begin{equation}
        \mathbb I_\omega(\lambda) =
        \begin{cases}
            1, & \lambda \in \omega, \\
            0, & \lambda \notin \omega.
        \end{cases}
    \end{equation}
We call this the \textbf{Dichotomic Response Function} assumption.}
\end{assumpt}
\begin{assumpt}\label{assumptclass4}
    \textit{For every measurable set $\omega \in \Sigma$ with a finite, non-zero measure (volume) $V(\omega) = \int_\omega d\lambda$, there exists a valid preparation procedure $P_\omega$. This preparation is represented by a uniform distribution over the set $\omega$:
    \begin{equation}
    \label{eq: indicator functions for prep}
        \mu_{P_\omega}(d\lambda) = \frac{\mathbb{I}_\omega(\lambda)}{V(\omega)} d\lambda,
    \end{equation}
    where $\mathbb{I}_\omega(\lambda)$ is the indicator function of the set $\omega$.
    We call this assumption the \textbf{Uniform Preparation} assumption.}
\end{assumpt}
We define a system which satisfies \textbf{Classical Assumptions~\ref{assumptclass1}, \ref{assumptclass2}, \& \ref{assumptclass3}, \ref{assumptclass4}} as \textbf{Classical}.

Notably, these assumptions imply that the system is tomographically complete, as per Ref.~\cite{Schmid2021}: the physical state is uniquely determined by the statistics of a set of measurements.

If we have the possibility of preparing a state in a perfectly precise state, the \textbf{Uniform Preparation} assumption (\textbf{Classical Assumption~\ref{assumptclass4}}) can be replaced by the following stronger assumption.
\begin{assumpt}\label{assumptclass5}
\textit{For every ontic state \(\lambda_0 \in \Omega\), there exists a valid preparation procedure \(P_{\lambda_0}\) that perfectly localises the system. Its associated ontic probability measure is the Dirac measure \(\delta_{\lambda_0}\) defined as 
\begin{equation}
    \delta_{\lambda_0}(\lambda)=
    \begin{cases}
        1\quad \text{if}\ \lambda=\lambda_0\\
        0 \quad \text{otherwise}
    \end{cases}
\end{equation}
We call this assumption \textbf{Point Preparation} assumption.}
\end{assumpt}
Note that the \textbf{Point Preparation} assumption implies the \textbf{Uniform Preparation} assumption, in the sense that a system that satisfies the former must also satisfy the latter. 

Moreover if we have the possibility of measuring a state in a perfectly precise state, the \textbf{Dichotomic Response Function} assumption (\textbf{Classical Assumption~\ref{assumptclass3}}) can be replaced by the following stronger assumption.
\begin{assumpt}\label{assumptclass6}
\textit{For every ontic state \(\lambda_0 \in \Omega\), there exists a valid measurement procedure \(M_{\lambda_0}\) that perfectly identifies the state of the system. Its associated ontic probability measure is the Dirac measure \(\delta_{\lambda_0}\) defined as 
\begin{equation}
    \delta_{\lambda_0}(\lambda)=
    \begin{cases}
        1\quad \text{if}\ \lambda=\lambda_0\\
        0 \quad \text{otherwise}
    \end{cases}
\end{equation}
We call this assumption \textbf{Point Measurement} assumption.}
\end{assumpt}
Note that, as before, the \textbf{Point Measurement} assumption implies the \textbf{Dichotomic Response Function} assumption, in the sense that a system that satisfies the former must also satisfy the latter. 

We define a system that satisfies \textbf{Classical Assumptions~\ref{assumptclass1}, \ref{assumptclass2},  \ref{assumptclass5}, \& \ref{assumptclass6}} as \textbf{Strongly Classical}.

We can now show that any \textbf{Classical} system is Spekkens preparation noncontextual.

\begin{theorem}
\label{theorem: classical system is univ NC}
    A \textbf{Classical} system is universally noncontextual.
\end{theorem}

Proof can be found in Appendix \ref{app: proof of theorem NC for classical system}.

We want to highlight that the assumptions made are clearly not valid for a quantum mechanical system, for the following reasons.
First, we assumed the ability to prepare a system uniformly over an arbitrarily defined measurable set $\omega\in \Sigma$. In quantum mechanics, Heisenberg's Uncertainty Principle sets a lower bound on the phase-space volume of any state. One cannot prepare a state $\rho$ that corresponds to an indicator function $\mathbb{I}_\omega$ if the volume of $\omega$ is smaller than $\hbar^n$.
Second, we assumed the existence of a complete set of dichotomic measurements. In quantum mechanics, most measurements are disturbing, and many states are non-orthogonal. Two different ontic states $\lambda_1,\lambda_2$ cannot be perfectly distinguished by a single measurement unless they correspond to orthogonal quantum states. This means that no single dichotomic measurement can in general distinguish between two different quantum states.

We focus now on the reverse implication - whether Spekkens noncontextuality implies \textbf{Classicality} (per our definition above).

To show that Spekkens noncontextuality doesn't imply \textbf{Classicality}, we would need a counterexample - a system which is not \textbf{Classical} (by our definition above), but which is still noncontextual by Spekkens' definition. An example of such a system is Gaussian quantum mechanics~\cite{brask2022gaussianstatesoperations}.
Gaussian quantum mechanics is the part of quantum mechanics that focuses on quantum states and operations that are fully characterised by their first and second moments (mean and covariance). These states are described as ``Gaussian'' because their Wigner functions~\cite{Wigner1932WignerFunctions} are positive Gaussian distributions. In Ref.~\cite{Bartlett2012GaussianQM} the authors reconstruct Gaussian quantum mechanics from Liouville mechanics and prove that Gaussian quantum mechanics is Spekkens noncontextual. However it is possible to show that Gaussian quantum mechanics violates both \textbf{Classical Assumptions \ref{assumptclass2} and \ref{assumptclass3}}. Proof of this can be found in Appendix \ref{app: Gaussian QM violates classical assumptions}. This shows that Spekkens noncontextuality does not imply classicality. It is important to notice also that again in Ref.~\cite{Bartlett2012GaussianQM} Gaussian quantum mechanics is also proved to be Bell-local and Kochen-Specker noncontextual, implying \textbf{Classicality} is also not a necessary condition for either of these properties.

Alternative definitions of classicality are obviously possible. For example, Ref.~\cite{Liu_2026} presents a new framework for Kochen-Specker contextuality, which is similar to the simplex embeddability framework for Spekkens contextuality. By their approach, a finite general system is only considered classical if it can be embedded into a classical system (a Boolean algebra) and its state is noncontextual. They define a state within this system as noncontextual if and only if it can be represented as a convex combination of deterministic states. They show then that contextuality (by their definition) is a sufficient condition for nonclassicality (by their definition), but it is not a necessary one. This means, despite using differing definitions, they arrive at the same hierarchy of implication as us - e.g., they show Kochen-Specker contextuality implies nonclassicality, but is not necessary for nonclassicality, which we can see from e.g., Gaussian quantum mechanics being nonclassical but Kochen-Specker noncontextual.

\subsection{Bell nonlocality}\label{subsec:Bell}
We now need to identify a sufficient criterion for a system to be inherently nonclassical. A natural contender for this is Bell nonlocality (i.e., ability of a system to violate a Bell inequality \cite{Bell1966}). Given Bell inequalities are formulated using our minimal classical assumptions, and so a system violating a Bell inequality must violate one of these assumptions, Bell-nonlocality provides us with a ``gold standard'' sufficient criterion for a system being nonclassical.
We give the derivation for one of the most widely-used Bell inequalities, the CHSH inequality~\cite{Clauser1969Proposed}, in Appendix~\ref{appendix: derivation of CHSH inequality}.

\begin{theorem}
\label{Bell's theorem}
A physical system is said to be nonlocal in the Bell sense if it can violate a Bell inequality, e.g., the CHSH inequality
\begin{equation}\label{eq:CHSH}
    |E(\vec{a},\vec{b})-E(\vec{a},\vec{d})|+|E(\vec{c},\vec{b})+E(\vec{c},\vec{d})|\leq2
\end{equation}
where $E(\vec{m},\vec{n})$ is the correlation function of two measurements performed by two parties, defined by
$E(\vec m,\vec n)=\sum_{A,B=\pm1} AB \, p(A,B \mid \vec m,\vec n)$,
with $A,B\in\{\pm1\}$ denoting the measurement outcomes for the two parties, in the measurement settings $\vec{m},\vec{n}$.
\end{theorem}

It is well known that quantum mechanics violates this theorem \cite{Freeman1972BellViolation, Aspect1982BellViolation, hensen2015loophole}, while every classical theory does not \cite{Fine1982},\cite{Brunner2014}. This allows us to define the violation of Bell's theorem as a sufficient condition for a system to be nonclassical.
However it is important to notice Bell-nonlocality cannot be a necessary condition for Kochen-Specker contextuality, because it can only be applied to composite systems: a single $n$-dimensional quantum system (e.g., a qutrit) cannot demonstrate Bell-nonlocality, even if it is Kochen-Specker contextual.

\subsubsection{Bell polytope}\label{subsubsec:BellPolytope}
The set of all possible behaviours which satisfy a Bell inequality is referred to as the Bell-local set~\cite{Scarani2019BellNLBook}.

\begin{definition}\label{Def:Local}
    \textit{A process is defined as Bell-local if it is of the form 
        \begin{equation}
\label{eq: Bell locality assumption}
p^{AB}_{\lambda}(\vec{a},\vec{b};\alpha,\beta)
=
p^{A}_{\lambda}(\vec{a};\alpha,*) \,
p^{B}_{\lambda}(\vec{b};*,\beta)
\end{equation}
    A behaviour is defined as Bell-local if it can be written as a convex combination of Bell-local processes. A behaviour that cannot be written that way is called Bell-nonlocal. 
    }
\end{definition}

From Def.~\ref{Def:Local}, the first property we can identify for the Bell-local set is convexity. This means that, for every pair of behaviours $\mathcal{P}_{LV,1}, \mathcal{P}_{LV,2}$ in the Bell-local set $\mathcal{L}$, their convex sum must also belong to the set, i.e. $\mathcal{P}=\alpha\mathcal{P}_{LV,1}+(1-\alpha)\mathcal{P}_{LV,2}\ \in \mathcal{L}\quad \forall \alpha \in [0,1]$.
The proof of this is based on Fine's representation of local behaviours \cite{Fine1982}: a behaviour $\mathcal{P}$ is local if and only if it is a convex mixture of local deterministic (LD) processes
\begin{equation}
\label{eq: Fine local behaviors}
p(a,b|x,y)=\sum_{j}\sum_kq_{jk}\delta_{a=f_j(x)}\delta_{b=g_k(y)}
\end{equation}
with $\lambda\equiv(j,k)$ and $\sum_{j,k}q_{jk}=1.$
Using this representation for $\mathcal{P}_{LV,1},\mathcal{P}_{LV,2}$ it is clear that also $\mathcal{P}$ has the same form, with $q_{jk}=\alpha q_{jk,1}+(1-\alpha)q_{jk,2}$. Since probabilities are bounded, the local set is also compact. This is due to the Heine-Borel theorem which says that any closed and bounded set over the field $\mathbb{R}$ is a compact set~\cite{Dugac1989HeineBorelTheorem}. A compact convex set can then be defined as the convex hull of its extremal points \cite{rudin1991functional}, i.e. those points that cannot be written as a convex sum of other points in the set. The extremal points of $\mathcal{L}$ are all the local deterministic behaviours and only those. This can be proved again using Fine's deterministic representation of a local behaviour.
Indeed it is easy to show that only behaviours which can be written in that form can be extremal. We suppose \textit{per absurdo} that a LD behaviour can be decomposed as convex combination 
\[
P_{LD,?} = \alpha P_{LD,1} + (1 - \alpha) P_{LD,2}
\]
with $0 < \alpha < 1$ and $P_{LD,1} \neq P_{LD,2}$. The latter two behaviours must differ for at least one pair of inputs $(x,y)$: for these inputs, there are $(\alpha_1, \beta_1) \neq (\alpha_2, \beta_2)$ such that
\[
P_{LD,1}(\alpha_1, \beta_1 \mid x,y) = 1,
\]
and
\[
P_{LD,2}(\alpha_2, \beta_2 \mid x,y) = 1.
\]
But then
\[
P_{LD,?}(\alpha_1, \beta_1 \mid x,y) = \alpha
\quad \text{and} \quad
P_{LD,?}(\alpha_2, \beta_2 \mid x,y) = 1 - \alpha.
\]
So $P_{LD,?}$ is not deterministic, contradicting the assumption that it was LD. This means that the local set is a bounded polytope embedded in $\mathbb{R}^n$. The $(n-1)$-dimensional hyperplanes that delimit it are called facets and are of finite number. To be a facet, a hyperplane must satisfy two properties. The first is that at least $n$ linearly independent extremal points must lie on the hyperplane. The second is that all the extremal points that do not lie on the hyperplane must be found on the same side of it. This is equivalent to saying that, if $v\cdot P=f$ is the equation of the points $P$ of the facet, then
\begin{equation}
    v\cdot P\leq f\quad \forall P \in \mathcal{L}
\end{equation}
This means a polytope is fully determined by listing either its extremal points or its facets. Bell inequalities are then the equations of the facets of the local polytope, which take the general form
\begin{equation}
    \mathcal{I}(P)\equiv \sum_{a,b,x,y}v_{abxy}p(a,b|x,y) \leq \mathcal{I}_L
\end{equation}

Given this definition of nonclassicality through Bell-nonlocality, we now need to identify the relation between Bell-(non)locality and (non)contextuality. Let us first do this for Kochen-Specker contextuality. Using Abramsky's sheaf-theoretic framework, we will show below that Bell-nonlocality implies Kochen-Specker contextuality~\cite{Abramsky2011Sheaf}.

\subsection{Relating Kochen-Specker contextuality and Bell-nonlocality}\label{subsec:KSandBell}

\subsubsection{Sheaf Theory}\label{subsubsec:Sheaf}

Sheaf theory is the part of mathematics that study sheaves: mathematical constructions concerned with passages from local properties to global ones. Sheaves have played a fundamental role in the development of many areas of modern mathematics, such as algebraic geometry, differential geometry, and topology~\cite{tennison1975sheaf, dimca2004sheaves, maclane1994sheaves}. Here, they present a way to study contextuality and nonlocality in a unified way. 

In this treatment we begin by defining a measurement scenario $(X,O,\mathcal{M})$ where $X$ is a finite set of measurements, $O$ is a finite set of outcomes and $\mathcal{M}$ is a family of measurement contexts $C \subseteq X$ representing sets of measurements that can be performed jointly. Let $P(X)$ denote the poset (partially ordered set) of subsets of $X$, ordered by inclusion, i.e. $X_{\leq (x)}=\{y \in X|\text{dim}(y)\leq \text{dim}(x)\}$ and regarded as a category in the usual way, i.e. $x\rightarrow y$ iff $x\subseteq y$ (there exists a morphism between two objects if and only if the first is a subset of the second). For each $C \in \mathcal{M}$ a distribution $O^C$ on the events is specified.

\begin{definition}
\textit{A presheaf on $P(X)$ is a functor}
\begin{equation}
    \mathcal{E}:P(X)^{\text{op}}\rightarrow \textbf{Set}
\end{equation}
\end{definition}
For each subset $C \subseteq X$, the presheaf $\mathcal{E}$ assigns a set $\mathcal{E}(C)$ and for each inclusion $C' \subseteq C$, a restriction map $\mathcal{E}(C)\rightarrow\mathcal{E}(C')$. The elements of $\mathcal{E}(C)$ are called sections over $C$ and they represent assignments of outcomes to the measurements in $C$. If $X$ itself is included as an object, a section over $X$ is called a global section. Moreover we can define the event presheaf as $\mathcal{E}(C):=O^C$, i.e. the set of all outcomes assignments to the measurements in $C$. A global section of this presheaf therefore corresponds to a deterministic assignment of outcomes to all measurements in $X$.

\begin{definition}
    \textit{A presheaf $\mathcal{E}$ on $P(X)$ is a sheaf if for any family of subsets $\{C_i\}$ with $C=\bigcup_iC_i$ and any family of sections $s_i\in \mathcal{E}(C_i)$ satisfying the compatibility condition:}
    \begin{equation}
        s_i|_{C_i\cap C_j}=s_j|_{C_i\cap C_j} \quad \forall i,j,
    \end{equation}
    \textit{there exists a unique section $s \in \mathcal{E}(C)$ such that $s|_{C_i}=s_i$ for all i}.
\end{definition}
This means that a sheaf is a presheaf with the additive condition that local data determine global data uniquely whenever they are mutually compatible.
To describe empirical models quantitatively, we specify a probability distribution over the assignments $\mathcal{E}(C)$ for each context $C \in \mathcal{M}$. To do so we define the distributor functor $D_R:\textbf{Set}\rightarrow\textbf{Set}$ which maps a set to the set of $R$-distributions over it, where a $R$-distribution is a function $d:X\rightarrow R$ which has finite support and with $\sum_{x \in X} d(x)=1$ and $R$ a semiring (a ring with no negative elements). The composition $D_R\mathcal{E}$ is of course a presheaf because the distributor functor goes from a set to a set so the composition $D_R\mathcal{E}:P(X)^{\text{op}}\rightarrow \textbf{Set}$ which is the definition of a presheaf.

\begin{definition}
    \textit{An empirical model $e$ over a measurement scenario $(X,O,\mathcal{M})$ is a family of distributions }
    \begin{equation}
        \{e_C\in D_R\mathcal{E}(C)\}_{C\in \mathcal{M}},
    \end{equation}
    \textit{assigning a distribution to each measurement context.}
\end{definition}
Such families must satisfy the compatibility condition to represent a physical model, which is
\begin{equation}
    e_{C}|_{C\cap C'}=e_{C'}|_{C\cap C'} \quad \forall C,C' \in \mathcal{M}
\end{equation}
which means that distributions associated with different contexts agree on their overlaps.

\subsubsection{Using Sheaf Theory to show Bell-Nonlocality implies Kochen-Specker Contextuality}\label{subsubsec:UsingSheafBellimpliesKS}

In the sheaf-theoretic framework, both contextuality and nonlocality can be characterised in terms of the existence of global sections of a distribution presheaf. In this approach, classicality is identified with the possibility of consistently extending locally observed statistics to a single global probabilistic model defined over all measurements simultaneously. Contextuality and nonlocality arise precisely when such an extension is impossible. Let $X$ denote the set of all measurements in a given scenario, and let $\mathcal{E}$ be the event sheaf assigning to each subset $C \subseteq X$ the set $\mathcal{E}(C)=O^C$ of outcome assignments for measurements in $C$. A global section $s \in \mathcal{E}(X)$ specifies a definite outcome for every measurement, independently of the context in which it is performed. Such global assignments therefore correspond to noncontextual hidden variable models, while a contextual hidden variable model is represented by a distribution $ b \in D_R\mathcal{E}(C)$. Each noncontextual assignment $s\in \mathcal{E}(X)$ induces a Dirac distribution $\delta_s \in D_R\mathcal{E}(X)$, defined by $\delta_s(s)=1$ and $\delta_s(s')=0$ for $s'\neq s$.
The restriction of $\delta_s$ to a context $C$ yields the deterministic distribution $\delta_s|_C\in D_R\mathcal{E}(C).$ Given a distribution $d$ over global sections, the empirical statistics for each context $C$ are recovered by
\begin{equation}
    e_C(s)=d|_C(s)=\sum_{\substack{s'\in \mathcal{E}(X)\\s'|_C=s}}d(s')=\sum_{s'\in \mathcal{E}(X)}\delta_{s'|C}(s)d(s')
\end{equation}
Since noncontextual hidden variables assign outcomes to each measurement independently of the context in which it appears, for each deterministic assignment $s$ and context $C$, the induced distribution factorises over individual measurements
\begin{equation}
    \delta_s|_C(s')=\prod_{x\in C}\delta_{s|\{x\}}(s'|\{x\})
\end{equation}
So one gets the following implication:
\begin{proposition}
\textit{The existence of a global section for an empirical model implies the existence of a non-contextual hidden-variable model which realises it.}
\end{proposition}

The problem of determining whether an empirical model admits a global section is reformulated as a linear feasibility question. Let $\mathcal{M}$ be a measurement cover of $X$, i.e. a collection of subsets of $X$ that represent maximal sets of jointly performable measurements. Consider the disjoint union of local sections, $\bigsqcup_{C\in \mathcal{M}} \mathcal{E}(C)$ and enumerate its elements as $\{s_i\}_{i=1}^p$, where $p$ is the total number of local sections across all contexts. Similarly we enumerate all global sections of $\mathcal{E}$ as $\{t_j\}_{j=1}^q \subseteq \mathcal{E}(X)$, where $q$ is number of global assignments. In general $p\neq q$. The incidence matrix \textbf{M} is the $p\times q$ Boolean matrix defined by
\begin{equation}
    \textbf{M}[i,j]=\left\{
\begin{array}{c}
1,\quad t_j|C=s_i\ (s_i\in \mathcal{E}(C))  \\
0,\quad \text{otherwise}
\end{array}
\right.
\end{equation}
So each element $M_{ij}$ is assigned taking the global assignment $t_j$, restricting it to the context to which $s_i$ belongs to and verifying if it matches the local assignment $s_i$.
Each column of \textbf{M} corresponds to a global assignment and each row identifies the set of global assignments compatible with a given local section. Equivalently, \textbf{M} provides a matrix representation of the restriction map
\begin{equation}
    \mathcal{E}(X)\longrightarrow\prod_{C \in \mathcal{M}}\mathcal{E}(C),\quad s\mapsto (s|_C)_{C\in \mathcal{M}}.
\end{equation}
The construction depends only on the measurement cover $\mathcal{M}$ and the event presheaf $\mathcal{E}$. Given an empirical model $\{e_C\}_{C\in\mathcal{M}}$ valued in $D_R$, each local section $s_i\in \mathcal{E}(C)$ is assigned a weight $e_C(s_i)$ in the underlying semiring $R$. If we consider all the contexts we get a vector $\textbf{v}\in R^p$ with components $\textbf{v}[i]=e_C(s_i)$. A global distribution is represented instead by a vector $\textbf{x}\in R^q$ whose entries assign weights to the global sections $\{t_j\}$. The requirement that the global distribution reproduces the empirical distributions on every context is encoded then by the linear system $\textbf{Mx= v}$. The existence of a solution \textbf{x} to this system characterises whether the empirical model admits a global section.
The following propositions holds:
\begin{proposition}
    \textit{Solutions of the augmented system} \textbf{Mx}=\textbf{v} \textit{in R are in one-to-one correspondence with global sections realising the empirical model.}
\end{proposition}
We can now relate hidden variables to global sections. Let $\Lambda $ denote a set of hidden variables. A hidden-variable model specifies, for each $\lambda \in \Lambda$ and each context $C \in \mathcal{M}$, a distribution $h_C^\lambda \in D\mathcal{E}(C)$, together with a context-independent distribution $h_\Lambda \in D(\Lambda)$ over hidden variables. For each fixed $\lambda$, the family $\{h_C^\lambda\}_{C\in \mathcal{M}}$ is required to satisfy the compatibility condition
\begin{equation}
    h_C^\lambda|_{C\cap C'}=h_{C'}^\lambda|_{C\cap C'}
\end{equation}
A hidden-variable model $h$ is said to realise an empirical model $e$ if the empirical probabilities are recovered by averaging over the hidden variables. Explicitly, for all $C\in \mathcal{M}$ and all $s \in \mathcal{E}(C)$,
\begin{equation}
    e_C(s)=\sum_{\lambda\in \Lambda}h_C^\lambda(s)\cdot h_\Lambda(\lambda).
\end{equation}
Moreover a hidden-variable model $h$ is said to be factorisable if, for every context $C\in\mathcal{M}$ and every section $s\in \mathcal{E}(C)$,
\begin{equation}
    h_C^\lambda(s)=\prod_{m\in C}h_C^\lambda|\{m\}(s|\{m\}).
\end{equation}
This condition states that the probability assigned to a joint outcome factors into probabilities assigned to individual measurement outcomes.
The following proposition can be proven to be right:
\begin{proposition}
\label{prop: Sheaf, global section = NC}
    \textit{Let $e$ be an empirical model defined on a measurement cover $\mathcal{M}$ for a distribution functor $D_R$. The following are equivalent:
    \begin{enumerate}
        \item $e$ has a realisation by a factorisable hidden-variable model.
        \item $e$ has a global section
    \end{enumerate}
    }
\end{proposition}
The proof can be found in Appendix \ref{app: Proof of NC = global section}.

An empirical model $\{e_C\}_{C\in \mathcal{M}}$ is said to be extendable to a global section if there exists a distribution $d \in D_R\mathcal{E}(X)$ such that $d|C=e_C$ for all $C\in \mathcal{M}$. 
In Bell scenarios, extendability to a global section is equivalent to Bell-locality, while in general measurement scenarios it coincides with Kochen-Specker noncontextuality. From this viewpoint, nonlocality appears as a subset of contextuality, related to the case in which spatially separated measurements are performed.

This demonstrates that Bell-nonlocality implies Kochen-Specker contextuality.

\subsubsection{Kochen-Specker Contextuality implies Bell-Nonlocality (in certain scenarios)}\label{subsubsec:KSimpliesBellish}

We would like now to demonstrate the reverse implication: that Kochen-Specker contextuality (in certain scenarios) implies Bell-nonlocality. Obviously in general this is not true, since Bell-nonlocality requires a multipartite system, while KS contextuality can arise in single systems (e.g., a single qutrit). Instead, what we want to prove is that, for a bipartite system where the individual parts are too small to manifest KS-contextuality (e.g., a system of multiple qubits), any KS-contextuality in the system implies the system is Bell-nonlocal. We will start with the example of a bipartite qubit system, i.e. a system in which one qubit is possessed by Alice, and another by Bob, who are spatially separated. We can consider two types of measurements on such a system: global and local. In a global measurement, both qubits are measured simultaneously as a single operation on the global state (e.g., in a Bell measurement, such as that used in teleportation, where the eigenstates of the measurement operator are the Bell states). In a local measurement however, one of the two qubits is measured locally (i.e., with a single-qubit measurement operator).

This distinction can be demonstrated using the Peres-Mermin square~\cite{PERES1990,Peres1991,Mermin1993}. The Peres-Mermin square is a square matrix whose elements are observables on a two-qubit Hilbert space $\mathcal{H}_2\otimes\mathcal{H}_2$:
\begin{equation}
\begin{bmatrix}
\sigma_z\otimes \mathds{1}_2 & \mathds{1}_2\otimes\sigma_z & \sigma_z\otimes\sigma_z\\
\mathds{1}_2\otimes\sigma_x & \sigma_x\otimes \mathds{1}_2 & \sigma_x\otimes\sigma_x\\
\sigma_z\otimes \sigma_x & \sigma_x\otimes\sigma_z & \sigma_y\otimes\sigma_y
\end{bmatrix}
\end{equation}
where $\mathds{1}_2$ is the identity operator for the two dimensional Hilbert space of one particle.
This array of measurements is one of the simplest examples which demonstrates KS-contextuality for a quantum-mechanical system. Every column/row forms a separate context, in the sense that all of its operators mutually commute. However, operators not in the same column/row do not necessarily commute, meaning we can define some noncontextuality inequality which together these operators violate.

Notably, the contextuality of the Peres-Mermin square is valid regardless of the initial state of the system's two qubits. However we can interpret the measurements of each observable in the last row and last column as being either two single-qubit local measurements (with four outcomes) or as a global (dichotomic) measurement. If they are performed as local measurements, one has in the last row a measurement of six incompatible single-qubit observables, which cannot form a context, making the scheme not more valid. Therefore, we can intuitively say that either a global (Bell-style) measurement, or a global entangled state, are necessary to prove the contextuality of the system.

We can formalise this intuition through the following proposition.

\begin{proposition}
\label{prop: state dep context implies Bell NL}
\textit{Any set of measurements exhibiting State-Dependent Contextuality or State-Independent Contextuality can be mapped to a Bell inequality that is violated by a maximally entangled state.}
\end{proposition}

Prop.~\ref{prop: state dep context implies Bell NL} can be demonstrated following Ref~\cite{Cabello2021ConvertingCtoNL}.\footnote{Note Ref.~\cite{Wright2023contextualityin} shows a similar result to Ref~\cite{Cabello2021ConvertingCtoNL} - that an entangled multi-qubit state can demonstrate Kochen-Specker contextuality using local (i.e., unentangled) measurements if and only if the state can be used (with appropriate local measurements) to violate a Bell inequality.}
We give a proof of this in Appendix \ref{Appendix: Cabello state dep context implies Bell NL}.

This shows that in certain scenarios (where maximal entanglement is present) Kochen-Specker contextuality implies Bell nonlocality.

\subsection{Relating Spekkens contextuality and Bell-nonlocality}\label{subsec:SpekkensandBell}

We now want to identify the relation between Spekkens contextuality and Bell nonlocality.

First of all we need the following result. Spekkens demonstrated that a maximally-mixed qubit state is preparation contextual~\cite{Spekkens2005}. However this result can be generalised to the following theorem~\cite{Banik2014}:
\begin{theorem}
\label{theorem: mixed qubit implies prep contextuality}
    \textit{Any mixed state of a qubit is preparation contextual.}
\end{theorem}
We prove this in Appendix~\ref{appendix: proof of theorem mixed qubit prep context}. 

\subsubsection{Spekkens Contextuality does not imply Bell-Nonlocality}\label{subsubsec:SpekkensnotimplyBell}

We showed in the previous section that Kochen-Specker state-independent contextuality implies Bell-nonlocality for any two maximally entangled qubit system. 
Here we show that Spekkens contextuality does not necessarily imply Bell-nonlocality for a two qubit system, and this can be shown using the Werner state. A Werner state \cite{Werner1989WernerState} is a particular bipartite quantum state which is invariant under all unitary operators in the form $U\otimes U\otimes...\otimes U$ for any unitary $U$ on a single subsystem. This means a Werner state $\rho _{AB}$ must satisfy
\begin{equation}
    \rho_{AB}=(U\otimes U)\rho_{AB}(U^\dagger\otimes U^\dagger)
\end{equation}
In the case of the two-qubit system it can be written as a convex combination of the maximally mixed state and a Bell state, for example
\begin{equation}
    W_{AB}=\alpha \ket{\Psi^-}\bra{\Psi^-}+\frac{1-\alpha}{4}I_{AB}
\end{equation}
where $-1/3\leq \alpha\leq 1$. Thus the two-qubit Werner states are separable for $\alpha\leq1/3$, and entangled for $\alpha>1/3$. Consider $1/3<\alpha\leq 1/2$. For these values of $\alpha$, the state is entangled, and the two density matrices for the two local qubits (each obtained by tracing out the other qubit) are each $\mathds{1}/2$ - i.e., maximally mixed. This implies by Theorem \ref{theorem: mixed qubit implies prep contextuality} that each of these local qubit states must be Spekkens preparation contextual.
However, the Werner state for these specific values of $\alpha$ can be described by a local hidden variable model~\cite{Werner1989WernerState}. This shows that Spekkens contextuality does not necessarily imply Bell-nonlocality.

\subsubsection{Geometrically proving that Bell-Nonlocality implies Spekkens Contextuality}\label{subsubsec:BellimpliesSpekkens}

Let us now consider the reverse implication: whether Bell-nonlocality implies Spekkens contextuality. 
Above, we saw that both Bell nonlocality and Spekkens contextuality can be treated geometrically. Section~\ref{subsubsec:BellPolytope} showed that Bell nonlocality consists of the set of behaviours which do not belong to the local polytope $\mathcal{L}$, while Section~\ref{subsubsec:simplex} showed that a necessary and sufficient condition for a scenario being Spekkens noncontextual is the embeddability of the Generalised Probabilistic Theory (GPT) representing the scenario in a simplex. In this Section, we link these two concepts to find the formal relations between them.

Since Bell nonlocality manifests in multipartite scenarios, we consider a bipartite system shared between two spacelike separated parties, Alice and Bob. To analyse this scenario in terms of simplex embeddability, we represent it as a GPT. We define the joint preparation of the system by a quantum state $\rho_{AB}$. 

Imagine Alice performs a local measurement $x$ yielding outcome $a$, described by a POVM element $E_{a|x}$. Bob similarly performs a local measurement $y$ yielding outcome $b$, described by a POVM element $E_{b|y}$. The joint measurement is represented by the tensor product effect $E_{ab|xy} = E_{a|x} \otimes E_{b|y}$. The joint probability distribution, or behaviour, is given by the Born rule:
\begin{equation}
    p(a,b|x,y) = \text{Tr}((E_{a|x} \otimes E_{b|y}) \rho_{AB})
\end{equation}

This allows us to define a formal bipartite GPT, $G_{\text{Bell}} = (V, \langle \cdot, \cdot \rangle, \Omega, \mathcal{E})$, where $\Omega$ is the state space containing the abstract state vector $\textbf{s}_P$ corresponding to the joint preparation $\rho_{AB}$, and $\mathcal{E}$ is the effect space containing the joint effects $\textbf{e}_{ab|xy}$ corresponding to the local operational structure $E_{a|x} \otimes E_{b|y}$.

In this GPT, the probability rule translates to the bilinear inner product:
\begin{equation}
    p(a,b|x,y) = \langle \textbf{e}_{ab|xy}, \textbf{s}_P \rangle_V
\end{equation}
The following theorem relates $G_{Bell}$ being simplex-embeddable and $p(a,b|x,y)$ belonging to the Bell-local polytope $\mathcal{L}$:
\begin{theorem}
    \textit{If a generalised probabilistic theory $G=(V,\langle\_,\_\rangle_V,\Omega,\mathcal{E})$ describing a bipartite prepare-measure experiment is simplex-embeddable and it is diagram-preserving in the sense defined in Ref.~\cite{Schmid2024structuretheorem}, then any behaviour $\mathcal{P}$ generated by $G$ belongs to the Bell-local polytope $\mathcal{L}$.}
\end{theorem}

\begin{proof}
    Consider a bipartite Bell scenario where a joint preparation $P$ is shared between two spacelike separated parties. The preparation corresponds to a state vector $\textbf{s}_P \in \Omega$. The parties perform local measurements $x$ and $y$, obtaining outcomes $a$ and $b$ respectively. This corresponds to a joint effect $\textbf{e}_{ab|xy} \in \mathcal{E}$. The behaviour generated by the GPT is given by the probability rule:
    \begin{equation}
        p(a,b|x,y) = \langle \textbf{e}_{ab|xy}, \textbf{s}_P \rangle_V
    \end{equation}

    By hypothesis, the GPT $G$ is simplex-embeddable. Therefore, by definition, there exist linear maps $\iota$ and $\kappa$ such that $\iota(\textbf{s}_P) \in \Delta_d$ and $\kappa(\textbf{e}_{ab|xy}) \in \Delta^*_d$, where $\Delta_d$ is a $(d-1)$-dimensional simplex. A fundamental geometric property of a simplex is that any point within its boundaries can be expressed as a unique convex combination of its $d$ vertices. Each vertex corresponds to an ontic state $\lambda$~\cite{Schmid2021}. Let us denote these vertices as $\textbf{v}_\lambda$, with $\lambda \in \{1, \dots, d\}$. We can therefore decompose the mapped state as:
    \begin{equation}
        \iota(\textbf{s}_P) = \sum_{\lambda=1}^d q_\lambda\textbf{v}_\lambda
    \end{equation}
    where $q_\lambda \geq 0$ and $\sum_\lambda q_\lambda = 1$. 
    Using the simplex-embeddability inner product condition $\langle\textbf{e},\textbf{s}\rangle_V=\langle\kappa(\textbf{e}),\iota(\textbf{s})\rangle_W$, we can rewrite the behaviour as:
    \begin{equation}
    \begin{split}
        p(a,b|x,y) = \langle \kappa(\textbf{e}_{ab|xy}), \iota(\textbf{s}_P) \rangle_W \\= \left\langle \kappa(\textbf{e}_{ab|xy}), \sum_{j=1}^d q_\lambda \textbf{v}_\lambda\right\rangle_W
    \end{split}
    \end{equation}
    By the linearity of the inner product space $W$, this becomes:
    \begin{equation}
        p(a,b|x,y) = \sum_{\lambda=1}^d q_\lambda \langle \kappa(\textbf{e}_{ab|xy}), \textbf{v}_\lambda \rangle_W
    \end{equation}

    The term $\langle \kappa(\textbf{e}_{ab|xy}), \textbf{v}_\lambda \rangle_W$ represents the probability of obtaining outcomes $a,b$ given settings $x,y$ and the ontic state $\lambda$. Because $\kappa(\textbf{e}_{ab|xy})$ resides in the dual hypercube $\Delta^*_d$ and $\textbf{v}_\lambda \in \Delta_d$, the definition of the dual guarantees that this inner product strictly represents a valid probability in $[0,1]$.

    In a standard bipartite Bell scenario, the joint measurements are composed of independent local operations, meaning the joint effect possesses a tensor product structure: $\mathbf{e}_{ab|xy}=\mathbf{e}_{a|x}\otimes\mathbf{e}_{b|y}$. This follows from the hypothesis of diagram preservation. Because the entire GPT is simplex-embeddable, it admits a classical noncontextual ontological model. In such a model, the embedding must preserve the compositional structure of the local operations. Consequently, the evaluation of the mapped joint effect $\kappa(\mathbf{e}_{a|x}\otimes\mathbf{e}_{b|y})$
    must factorise into the product of the marginal evaluations. We can therefore write:
\begin{equation}
\langle \kappa(\textbf{e}_{ab|xy}), \textbf{v}_\lambda \rangle_W = p_A(a|x,\lambda) p_B(b|y,\lambda)
\end{equation}
    where $p_A(a|x,\lambda)$ and $p_B(b|y,\lambda)$
    are valid probabilities in $[0,1]$ representing the local response functions.
    
    As established by Fine's theorem \cite{Fine1982}, any such factorisable local probability distribution can be represented as a convex mixture of local deterministic processes. Thus, for each $\lambda$, we can expand the local response functions into local deterministic behaviours:
    \begin{equation}
    \begin{split}
        p_A&(a|x,\lambda) p_B(b|y,\lambda) =\\
        &\sum_{j,k} P(j,k|\lambda) \delta_{a=f_j(x)} \delta_{b=g_k(y)}
    \end{split}
    \end{equation}
    where $P(j,k|\lambda)$ is a valid probability distribution over the deterministic assignments $j$ and $k$. Substituting this back into the geometric expression for our behaviour yields:
    \begin{equation}
        p(a,b|x,y) = \sum_{\lambda=1}^d q_\lambda \sum_{j,k} P(j,k|\lambda) \delta_{a=f_j(x)} \delta_{b=g_k(y)}
    \end{equation}

    By defining a new aggregate probability distribution $q_{jk} \equiv \sum_{\lambda=1}^d q_\lambda P(j,k|\lambda)$, we obtain:
    \begin{equation}
        p(a,b|x,y) = \sum_{j,k} q_{jk} \delta_{a=f_j(x)} \delta_{b=g_k(y)}
    \end{equation}
    with $\sum_{j,k} q_{jk} = 1$. This exactly matches the deterministic representation of a local behaviour. Therefore, the overall behaviour $\mathcal{P}$ is geometrically constrained to be a convex mixture of local deterministic processes, proving that $\mathcal{P} \in \mathcal{L}$, the Bell local polytope.
\end{proof}

This result proves that in a bipartite system Spekkens noncontextuality implies Bell locality - or equivalently via contraposition, that in a bipartite system, Bell nonlocality implies Spekkens contextuality.

(In Ref.~\cite{Mukherjee2024} an even stronger result is reported. The authors show that simplex embeddability in a bipartite system implies unsteerability, which is a sufficient condition for Bell-locality.)

The result obtained is in agreement with the result obtained in Ref.~\cite{Wright2023InvertibleMap}. In that work the authors show the existence of an invertible map between correlations in a bipartite Bell scenario and Spekkens (non)contextual behaviours in a specifically-constructed family of prepare-and-measure scenarios. The invertibility of the map only holds for contextuality scenarios of a particular structural type (their "remote-preparation" scenarios, satisfying their criteria I/II on preparation equivalences), while this correspondence breaks for more general contextuality scenarios as we expect.

\subsection{Summary of Relations}\label{subsec:RelationSummary}

See Table~\ref{table:Relations} for a summary of the chain of implications which we demonstrated above between (Non)Classicality, Bell-(Non)locality, Kochen-Specker (Non)Contextuality, and Spekkens (Non)Contextuality. Fig.~\ref{fig:twoplots} illustrates these implications using set diagrams. Fig.~\ref{fig.hierarchy} shows how these two sets of implications intersect, and gives different example systems for each tier of intersection.

\renewcommand{\arraystretch}{1.5}
\begin{table*}
\centering
\begin{tabular}{||c|c|c|c|c||}
\hline\hline
& \multicolumn{4}{c||}{\textbf{Property $Y$}}\\
\hline\hline
\textbf{Property $X$} 
& Spekkens Contextuality 
& KS Contextuality 
& Bell Non-Locality 
& Non-Classicality\\
\hline\hline

\multirow{2}{*}{Spekkens Contextuality} 
& $X \Leftrightarrow Y$ 
& $X \Leftarrow Y$ 
& $X \Leftarrow Y$ 
& $\neg X \Leftarrow \neg Y$ \\
& $\neg X \Leftrightarrow \neg Y$ 
& $\neg X \Rightarrow \neg Y$ 
& $\neg X \Rightarrow \neg Y$ 
& $ X \Rightarrow Y$ \\
\hline

\multirow{2}{*}{KS Contextuality} 
& $X \Rightarrow Y$ 
& $X \Leftrightarrow Y$ 
& $X \Leftarrow Y$ 
& $X \Rightarrow Y$ \\
& $\neg X \Leftarrow \neg Y$ 
& $\neg X \Leftrightarrow \neg Y$ 
& $\neg X \Rightarrow \neg Y$ 
& $\neg X \Leftarrow \neg Y$ \\
\hline

\multirow{2}{*}{Bell Non-Locality} 
& $X \Rightarrow Y$ 
& $X \Rightarrow Y$ 
& $X \Leftrightarrow Y$ 
& $X \Rightarrow Y$ \\
& $\neg X \Leftarrow \neg Y$ 
& $\neg X \Leftarrow \neg Y$ 
& $\neg X \Leftrightarrow \neg Y$ 
& $\neg X \Leftarrow \neg Y$ \\
\hline

\multirow{2}{*}{Non-Classicality} 
& $\neg X \Rightarrow \neg Y$ 
& $X \Leftarrow Y$ 
& $X \Leftarrow Y$ 
& $X \Leftrightarrow Y$ \\
& $X \Leftarrow Y$ 
& $\neg X \Rightarrow \neg Y$ 
& $\neg X \Rightarrow \neg Y$ 
& $\neg X \Leftrightarrow \neg Y$ \\
\hline\hline
\end{tabular}
\caption{Restricted summary of relations between Spekkens contextuality, Kochen-Specker contextuality, Bell non-locality, and \textbf{Nonclassicality} (i.e., violation of \textbf{Classical Assumptions~\ref{assumptclass1}, \ref{assumptclass2}, \& \ref{assumptclass3}}). Note here we have a full hierarchy: \textbf{Classicality} $\Rightarrow$ Spekkens Noncontextuality $\Rightarrow$ Kochen Specker Noncontextuality $\Rightarrow$ Bell Locality; or Bell-Nonlocality $\Rightarrow$ Kochen Specker Contextuality $\Rightarrow$ Spekkens Contextuality $\Rightarrow$ \textbf{Nonclassicality}. Obviously the inverse is not true - e.g., Spekkens Contextuality $\nRightarrow$ Kochen-Specker contextuality.}
\label{table:Relations}
\end{table*}

\begin{figure}
    \centering
    \begin{subfigure}{0.22\textwidth}
        \includegraphics[width=\linewidth]{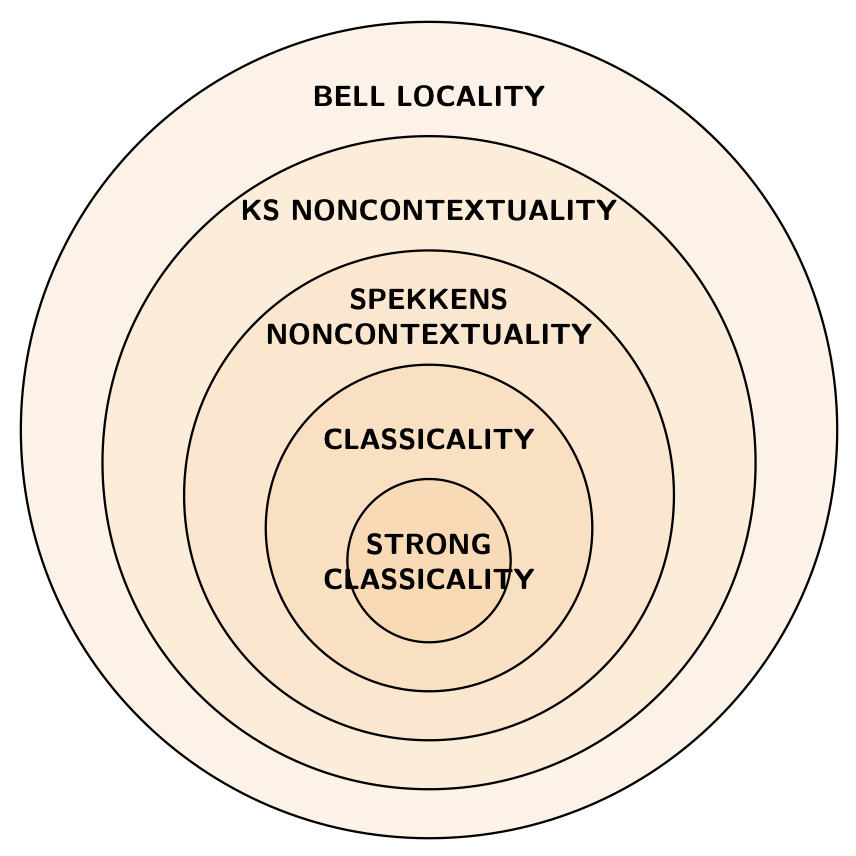}\caption{}
        \label{fig:plot1}
    \end{subfigure}
    \hfill
    \begin{subfigure}{0.22\textwidth}
        \includegraphics[width=\linewidth]{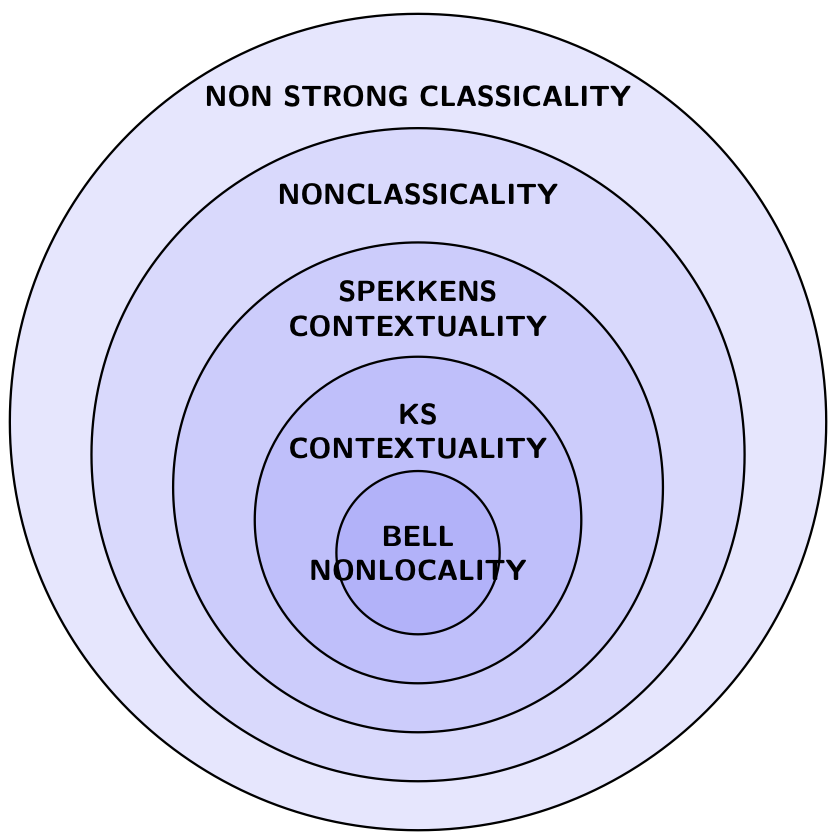}\caption{}
        \label{fig:plot2}
    \end{subfigure}
    
    \caption{Relations between Fig.~\ref{fig:plot1} Bell-Nonlocality, Kochen-Specker contextuality, Spekkens contextuality \textbf{Nonclassicality} (i.e., violation of \textbf{Classical Assumptions~\ref{assumptclass1}, \ref{assumptclass2}, \& \ref{assumptclass3}}), and \textbf{Non Strong Classicality} (i.e., violation of \textbf{Classical Assumptions~\ref{assumptclass1}, \ref{assumptclass4}, \& \ref{assumptclass5}}) of a system, and Fig.~\ref{fig:plot2} \textbf{Strong Classicality} (i.e., \textbf{Classical Assumptions~\ref{assumptclass1}, \ref{assumptclass4}, \& \ref{assumptclass5}}), \textbf{Classicality} (i.e., \textbf{Classical Assumptions~\ref{assumptclass1}, \ref{assumptclass2}, \& \ref{assumptclass3}}), Spekkens noncontextuality, Kochen-Specker noncontextuality and Bell-locality of a system.}
    \label{fig:twoplots}
\end{figure}

\begin{figure}
    \centering
    \includegraphics[width=1.0\linewidth]{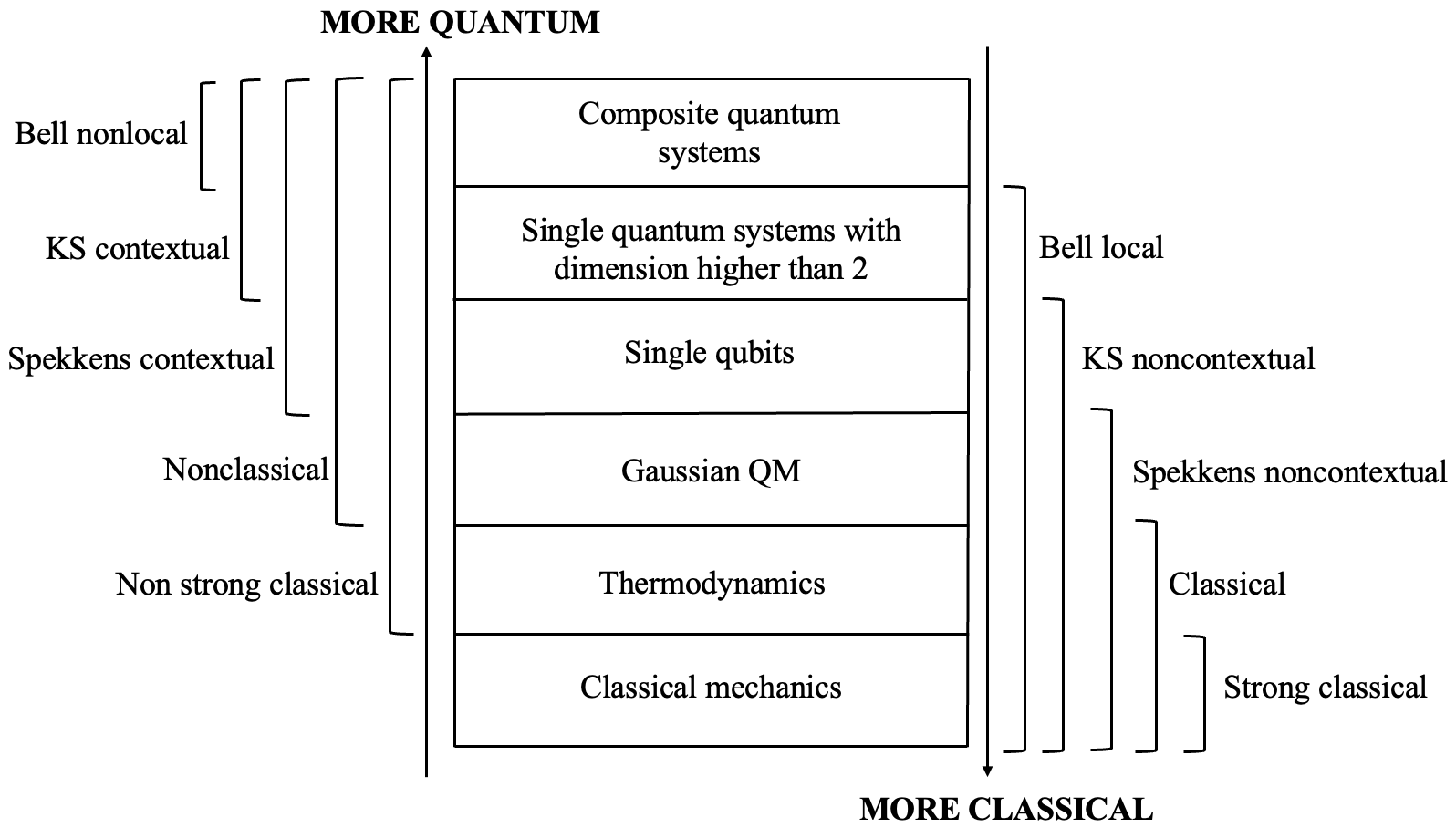}
    \caption{Illustration of the hierarchies shown in Fig.~\ref{fig:twoplots}, with example systems given for each step on this hierarchical ladder.
    }\label{fig.hierarchy}
\end{figure}

\subsection{Necessary vs Sufficient Conditions (or Provably Classical vs Provably Nonclassical)}

The results we obtained above allow us to link the notions of Kochen-Specker and Spekkens (non)contextuality to the idea of a (non)classicality of a system. In the literature there is no a clear and universally accepted distinction between classical and nonclassical systems. We propose that Kochen-Specker contextuality should be viewed as giving a sufficient criterion for nonclassicality (motivated through structural similarity to Bell-nonlocality). Similarly, while a system being Spekkens noncontextual does not necessarily imply that system is \textbf{Classical} in the way we define in Section~\ref{subsec: A classical system must be Spekkens noncontextual}, we would argue Spekkens noncontextuality should be viewed as giving a sufficient criterion for a more general notion of classicality (motivated through structural similarity to the tomographic completeness which underpins classical physics). The two notions - ``Kochen-Specker contextuality''-as-nonclassicality and ``Spekkens noncontextuality''-as-classicality - in our view are not competing; they are complementary. Which one should be used to evaluate a given scenario simply depends whether we’re looking for sufficient proof of classicality or sufficient proof of nonclassicality.  

Further, we showed above that Kochen-Specker contextuality is necessary but not sufficient for Bell-nonlocality. However Bell nonlocality also requires a system be multipartite. Sheaf theory allows us to see though that the underlying structure of Bell-local and Kochen-Specker noncontextual systems is the same - extendability of an empirical model to a global section. This makes both Bell-nonlocality and Kochen-Specker contextuality the inability to extend the empirical model to a global section - the only difference being whether, when embedded in a spacetime, that empirical model is formed of spatially-partitioned elements. Therefore, in this sense Kochen-Specker contextuality serves as a suitable sufficient condition for the same intuitive notion of nonclassicality which Bell-nonlocality appeals to - just one which doesn't require our system to be spatially partitioned. 

An example of this complementary approach, analysing a phenomenon using both notions of (non)contextuality, can be seen in Ref.~\cite{hance2026noncontextualversuscontextualinterferometry}. Here, it is shown that two-path interferometry (the Elitzur-Vaidman Bomb Tester and its extensions) is Spekkens noncontextual, so in some sense fundamentally ``classical''; but Hofmann's three-path interferometer~\cite{Hofmann_2026} is Kochen-Specker contextual, so in some sense fundamentally ``non-classical''. Through this, the phenomenology of interference is shown to have both classical and nonclassical aspects. This shows the efficacy of testing a given physical phenomenon using both notions: Spekkens noncontextuality to show whether something behaves in a fundamentally classical way; and KS contextuality to show whether it behaves in a fundamentally nonclassical way.

\section{Discussion}\label{sec:discussion}

This paper investigated the relation between the two main notions of contextuality present in literature: Kochen-Specker and Spekkens contextuality. After briefly recapping these definitions, we compared them and highlighted their differences. We then proved that Spekkens contextuality is a generalisation of Kochen-Specker contextuality, insofar as every system which is Kochen-Specker contextual is also Spekkens contextual (though, even for sharp projective measurements, not every system which is Spekkens contextual is Kochen-Specker contextual). After that, we linked these two notions to the concept of classicality/nonclassicality of a system. We defined sufficient conditions to identify a system as classical, and proved that a system meeting these conditions must be Spekkens noncontextual. However, not every Spekkens noncontextual system meets these conditions. 

We introduced then the concept of Bell nonlocality, and linked it to the two notions of contextuality. Using Abramsky's sheaf theoretic approach, we showed that Kochen-Specker contextuality is a generalisation of Bell nonlocality - both original from an inability to extend an empirical model to a global section. We also showed that in certain scenarios (with maximally entangled states) Kochen-Specker contextuality implies Bell nonlocality. Finally, we showed that Bell nonlocality implies Spekkens contextuality (but not the inverse - i.e., Spekkens contextuality doesn't necessarily imply Bell-nonlocality).

This allowed us to demonstrate a hierarchy of implication between these concepts - showing that \textbf{Classicality} $\Rightarrow$ \textbf{Spekkens Noncontextuality} $\Rightarrow$ \textbf{Kochen-Specker Noncontextuality} $\Rightarrow$ \textbf{Bell-Locality}; or by contraposition \textbf{Bell-Nonlocality} $\Rightarrow$ \textbf{Kochen-Specker Contextuality} $\Rightarrow$ \textbf{Spekkens Contextuality} $\Rightarrow$ \textbf{Nonclassicality}. Through this hierarchy, and examination of the physical meaning of the assumptions underpinning each notion, we showed that Kochen-Specker contextuality provides a good sufficient condition for the same notion of nonclassicality which underpins Bell-Nonlocality (as illustrated by the similarity of the sheaf-theoretic representation of their conditions); while Spekkens noncontextuality provides a suitable sufficient condition for a generalisation of the idea of classicality (coming from their shared reliance on ideas of tomographic completeness). This allows us to say Spekkens noncontextual system is ``provably classical'', while a Kochen-Specker contextual system is ``provably nonclassical''.

Ref.~\mbox{\cite{Khrennikov_2019}} discussed the idea that Bell nonlocality is a consequence of the property of incompatibility of quantum mechanical observables. Under this view means there is no requirement for any ``spooky action at a distance'' as Einstein thought. The sheaf theoretic framework we discuss in \mbox{\cref{subsec:KSandBell}} follows this idea. Bell nonlocality there arises as the impossibility of consistently extending locally observed statistics to a single global probabilistic model defined over all measurements simultaneously and this comes from the incompatibility of some quantum mechanical observables. Further, by showing the link between Kochen-Specker contextuality (which inherently results from incompatibility) and Bell-nonlocality, the sheaf-theoretic model reinforces the conclusions of Ref.~\mbox{\cite{Khrennikov_2019}}.

Future work will aim to integrate into this hierarchy other notions of nonclassicality which have been proposed in the literature - such as Leggett-Garg macrorealism violation \cite{lgi}, steerability~\cite{Wiseman_2007,Uola2020Steering,Mukherjee2024}, and quasiprobability non-Kolmogorovianity \cite{Ferrie2011Quasiprobabilities,Veitch2012Resource,Gherardini2024QuasiprobsTutorial} (e.g., Wigner negativity~\cite{Kenfack_2004}, Kirkwood-Dirac negativity/imaginarity~\cite{Kirkwood1933KD,Dirac1945KD,Halpern2018KD,Hance2024CFBAIG,Arvidsson-Shukur_2024KD}, and the linked notions of weak value anomalousness~\cite{Pusey2014Anomalous,Hance_2023} and Bargmann Invariant phases~\cite{Wagner_2024Bargmann}).

Further, based on suggestions from the literature \cite{Schmid2024structuretheorem,rossi2025typicalcontextuality}, in future we will investigate whether we can understand systems which fall into the gap between the two notions of contextuality we discuss - i.e., systems which are contextual by Spekkens' notion but noncontextual by Kochen and Specker's notion - as systems which \emph{could} be represented classically, but at the expense of certain arguably-desirable features (linearity and diagram-preservation).

\textit{Acknowledgements -} We thank Martin Pl\'avala, Jan-\AA ke Larsson, John Selby, and David Schmidt, for useful discussions which prompted this paper. JRH acknowledges support from a Royal Society Research Grant (RG/R1/251590), and from their EPSRC Quantum Technologies Career Acceleration Fellowship (UKRI1217).

\appendix
\begin{appendices}

\onecolumngrid
    
\section{Kochen-Specker contextuality for non-perfectly commuting observables}
\label{appendix: KS contextuality for non-perfectly commuting observables}

Following Ref.~\cite{Gunhe2010compatibility}, we here show how it is possible to apply the notion of Kochen-Specker contextuality to non-perfectly commuting observables.

The idea is to estimate the amount of disturbance introduced by the subsequent measurement of two non-``perfectly compatible'' measurements, considering the hypothesis that it is cumulative, which is the only assumption necessary.

We consider a hidden-variable model describing the probabilities of all possible sequences $S_{AB}=\{A,B,AA,AB,BB,...\}$ of two dichotomic observables $A,B$. The outcome probabilities are denoted as $p[\pm,A]$, $p[\pm,B]$ for single measurements and $p[\pm\pm,AB]$ for sequences of measurements. We can include the case of discarded outcomes, represented by $(\bullet)$. If A and B are compatible observables, it must necessarily be that $p[+\bullet -,BAB]=0$ because $A$ must respect the eigenspace of $B$ and then the state of the system after $A$ will remain in the eigenstate of $B$ associated to the positive eigenvalue. At the end we define $p[(+|A)\&(+|B)]$ as the the probability that the first measurement gives the outcome $+1$, in the case both of a measurement of $A$ and of a measurement of $B$.

For this model the following inequality holds:
\begin{equation}
    p[(+|A)\&(+|B)]\leq p[++,AB]+p[(+|A)\&(\bullet-|AB)]
\end{equation}
This is because the probability of measuring $A$ and $B$ both in the positive eigenvalue must be less than the probability that, if measured sequentially, $A$ and $B$ both give positive eigenvalue or the measure of $A$ flips the result of the measurement of B. We can define the correlator between $A$ and $B$ as $\langle AB\rangle=\sum_{ab=\pm1}ab\ p[(a|A)\&(b|B)]$ which is equal to:
\begin{equation}
    \langle AB\rangle=1-2p[(+|A)\&(-|B)]-2p[(-|A)\&(+|B)]
\end{equation}
The correlation measured during a real experiment instead is defined as $\langle A_1B_2\rangle=\sum_{ab=\pm1}ab\ p[(ab|AB)]$ and it will be different from $\langle AB\rangle$ if $A,B$ are incompatible. We can define the probability of flipping of the observable $B$ by the measurement of $A$: 
\begin{equation}
    p^{flip}[AB]:=p[(+|B)\&(\bullet-|AB)]+p[(-|B)\&(\bullet+|AB)]
\end{equation}
Thus we can bound the correlator of $A$ and $B$ as follows:
\begin{equation}
    \langle A_1B_2\rangle-2p^{flip}[AB]\leq \langle AB\rangle\leq\langle A_1B_2\rangle+2p^{flip}[AB]
\end{equation}
which is equal to:
\begin{equation}
    |\langle AB \rangle-\langle A_1 B_2 \rangle|\leq2p^{flip}[AB]
\end{equation}
which tells that the true underlying correlation ($\langle AB \rangle$) cannot be further away from the measured experimental correlation ($\langle A_1B_2 \rangle$) than twice the probability of flipping when the two measurements are performed sequentially.

To obtain only experimentally testable variables we need to bound also $p^{flip}$. We introduce $p^{err}[BAB]:=p[+\bullet -|BAB]+p[-\bullet +|BAB]$ which corresponds to the probability of flipping the value of $B$ in a sequence of three measurements with an intermediate measurement of $A$. This probability is experimentally measurable. We apply then the condition of cumulative noise:
\begin{equation}
  p[(\pm|B)]\land(\bullet \mp|AB)]\leq p[(\pm|B)]\land(\pm\bullet \mp|BAB)]
  =p[\pm\bullet \mp|BAB)
\end{equation}
This implies that $p^{flip}[AB]\leq p^{err}[AB]$. In such a way we have a bound also for $p^{flip}$ through $p^{err}$ which is experimentally testable.

Thanks to this it is possible to obtain experimentally testable inequalities which include also non-perfectly compatible observables.

\section{Relaxing the condition of Tomographic Completeness}
\label{appendix: relaxing Tomography}
Recently Ref.~\cite{pusey2019contextualityaccesstomographicallycomplete} showed that the Tomographic Completeness assumption is not necessary to demonstrate a system is Spekkens noncontextual. This work showed that it is possible to formulate tests of contextuality that still work even if there are a certain number of unknown procedures in the tomographically complete set. The authors demonstrated the following theorem.
\begin{theorem}
    For any $k \in \mathbb{N}$ there exists $2^k$ preparations and measurements, with statistics compatible with a qubit model, that would require $k$ measurements in a tomographically complete set for a preparation noncontextual model.
\end{theorem}
This means that the dependence of known preparations and measurements on unknown measurements is exponential.
After that the authors provide an algorithm which determines if a noncontextual ontological model exists or not, given a tomographically incomplete set of measurements.

Consider the state space of functions from (known) measurements to outcomes, i.e. deterministic assignments \(\lambda(M)=k\). For a preparation \(P\) we define the assignment polytope \(\Delta_P\) of distributions consistent with \(P\)'s statistics, i.e. distributions \(\mu(\lambda)\) such that for all measurements \(M\):
\begin{equation}
\label{eq: linearity Pusey non tomo completeness}
\sum_{\lambda}\delta_{k\lambda(M)}\,\mu(\lambda)=P^e(k\mid P,M).
\end{equation}

Clearly this is non-empty because, we have e.g. $\mu\in\Delta_P$ where
$\mu(\lambda)=\prod_M P^e(\lambda(M)\mid P,M)$.
\(\Delta_P\) is defined by positivity and the linear equality Eq.~\eqref{eq: linearity Pusey non tomo completeness}. The first step of the algorithm is to convert this polytope into a list of vertices.

The second step is, for every pair of disjoint subsets \(\{P_{i_1},P_{i_2},\ldots\}\) and \(\{P_{j_1},P_{j_2},\ldots\}\), to use a simple linear program (based on the lists of vertices) to check if corresponding convex hulls
$\mathrm{Con}(\{\Delta_{P_{i_1}},\Delta_{P_{i_2}},\ldots\})
\quad\text{and}\quad
\mathrm{Con}(\{\Delta_{P_{j_1}},\Delta_{P_{j_2}},\ldots\})$
intersect.

The authors then proved the following theorem.
\begin{theorem}
    If the above algorithm finds there are no intersections, there is no Spekkens noncontextual model.
\end{theorem}

This work allows us to relax the necessity of the condition of Tomographic Completeness when demonstrating Spekkens contextuality. Although determining if two preparations are operationally equivalent still requires access to a tomographically complete set, this result shows that an incomplete set can be enough to determine whether there exist operationally equivalent preparations and so to show whether the system cannot be represented noncontextually. This criterion can then be applied to scenarios in which no tomographically-complete set of measurements is available.

\section{POVMs extension for KS contextuality}\label{sec: appendix POVMs for KS}
Here we provide a brief view of how it is possible to extend the notion of KS contextuality to include also nonprojective measurements.

Positive Operator-Valued Measures (POVMs) are measurements whose elements do not need to be orthogonal or normalised, unlike sharp measurements or PVMs (projection valued measures). 
This means two different outcomes of the same POVM can correspond to eigenstates of two non-commuting observables. If we extend the standard KS definition of contextuality from PVMs, where it normally applies, to POVMs, this implies a single reproducible physical apparatus can yield outcomes belonging to entirely different measurement contexts, which seems peculiar. However to make sense of this, the authors of Ref.~\cite{hance2025externalquantumfluctuationsselect} suggest  applying Neumark's Dilation Theorem \cite{gelfand1943imbedding} which allows us to consider every POVM applied on a physical system as a PVM applied on a composite system comprising the physical system plus the measurement apparatus. Since different outcomes correspond to different macroscopic states in such a Neumark-extended frame, the different outcome states must be orthogonal in this extended Hilbert space. In this way it is possible to generalise the Kochen-Specker notion of contextuality to cover POVMs measurements. 

\section{Contextuality of a Qubit}\label{appendix:qubit}

There is a vast literature on qubits and their properties, including on whether they display contextuality. First Bell~\cite{Bell1966}, and more recently Spekkens~\cite{Spekkens2005}, and Johansson and Larsson~\cite{Johansson2019QSL}, demonstrated that a qubit is always representable using a noncontextual hidden variable model. Specifically, Johansson and Larsson extended Spekkens' toy model \cite{Spekkens2007Evidence} (a classical model for the partial simulation of quantum mechanical systems) to produce Quantum Simulation Logic (QSL), which mimics some key features of quantum computing despite being completely classical (formed of two potentially-random bits).

In QSL, a qubit is represented by a tuple $(x_0,p_0)$ where $x_0$ and $p_0$ are both classical bits. The first one is called the computational (or presence) bit, while the second the phase bit. For example the eigenstates of the $Z$ observable are represented by: $(0,R)\sim \ket{0}$ and $(1,R)\sim \ket{1}$ where $R\in\{0,1\}$ is a random variable. Inverting the assignments of the computational and phase bits gives the eigenstates of the $X$ observable instead. In this model the transition probability from the state $\ket{\psi}$ to $\ket{\phi}$, which in standard quantum mechanics is $|\braket{\phi}{\psi}|^2$, is instead the complement of the Kolmogorov distance between the distributions $P$ and $Q$ which describe the two states:
\begin{equation}
    F^2(P,Q)=1-\delta(P,Q)=1-\frac{1}{2}\sum_{x\in\Omega}|P(x)-Q(x)|
\end{equation}
Unitary transformations are represented by classical reversible logic operations on the bits composing a QSL state. For example the $X$ gate is represented by: $\mathcal{X}(x_0,p_0)=(x_o\oplus 1,p_0)$.

A projective measurement returns the corresponding bit (or a combination of the two if $Y$ measurement is performed). Specifically, a projective measurement $\mathcal{Z}$ returns the computational bit and randomises the phase bit, while a projective measurement $\mathcal{X}$ returns the phase bit and randomises the computational bit. $\mathcal{Y}$ instead returns the parity of the computational and phase bit and then randomises both the computational and phase bits while preserving parity.

This model is able to simulate in a Turing machine any single-qubit quantum circuit while requiring at most a constant overhead in resources. This model is noncontextual by construction since it simultaneously assigns values to all observable quantities, and these values do not change regardless of the measurement or measurement context that we use to retrieve them. This demonstrates then that a qubit can always be represented by a noncontextual hidden-variable model.

Spekkens's toy model, and so quantum simulation logic, are  therefore Kochen-Specker noncontextual, but we can show they are also Spekkens noncontextual. A model is defined as Spekkens noncontextual if operationally equivalent procedures have identical ontological representations. In the QSL model a procedure is represented by classical reversible logic operations on the two bits composing the QSL state, while measurements are operations which return in general a combination of the two bits. Spekkens defines two preparation procedures as being operationally equivalent if they yield identical statistics for all possible measurements. We can define a function $f:\Omega \times \Omega\rightarrow\Omega\times\Omega$ and a function $g:\Omega \times \Omega\rightarrow\Omega$ where $\Omega$ is the set of values that $x_0,p_0$ can assume. The first function represents the action of a transformation on the QSL state, while the second represents the action of a measurement on the QSL state.
Two operationally equivalent procedures $f,f'$ are then two procedures such that:
\begin{equation}
    g(f(x_0,p_0))=g(f'(x_0,p_0))\quad \forall g
\end{equation}
while two measurement outcomes are equivalent if:
\begin{equation}
    g'(f(x_0,p_0))=g(f(x_0,p_0))\quad \forall f
\end{equation}
The ontological representation of a procedure is a distributions over a space of hidden physical states $\lambda\in\Lambda$ and the probability of an outcome is given by Eq.~\eqref{equation ontological procedures representation}.

A preparation procedure $f$ is represented by a probability distribution $\mu_f(\lambda)$ over $\Gamma$, while a measurement $g$ with outcome $k$ is represented by a response function $\xi_{g,k}(\lambda)$. The probability of obtaining outcome $k$ given the preparation $f$ and the measurement $g$ is therefore
\begin{equation}
    p(k|f,g)=\sum_{\lambda\in\Gamma}
    \xi_{g,k}(\lambda)\,\mu_f(\lambda).
\end{equation}

If two preparation procedures $f$ and $f'$ are operationally equivalent,
then by definition
\begin{equation}
    p(k|f,g)=p(k|f',g)
    \quad \forall g,k.
\end{equation}
Hence
\begin{equation}
    \sum_{\lambda}
    \xi_{g,k}(\lambda)\,\mu_f(\lambda)
    =
    \sum_{\lambda}
    \xi_{g,k}(\lambda)\,\mu_{f'}(\lambda)
    \quad \forall g,k.
\end{equation}

In the QSL model (as in Spekkens' toy model) the ontological representation is deterministic, in the sense that the ontic state completely determines the outcome of any measurement, i.e.
$\xi_{g,k}(\lambda)\in\{0,1\}$ and for each $\lambda$ there exists a unique
$k$ such that $\xi_{g,k}(\lambda)=1$.
Since the set of measurements is complete over the ontic state space, equality of all operational statistics implies
\begin{equation}
    \mu_f(\lambda)=\mu_{f'}(\lambda)
    \quad \forall \lambda.
\end{equation}

Therefore operationally equivalent preparation procedures have identical
ontological representations. An analogous argument holds for measurements:
if two measurements are operationally equivalent, equality of the statistics
for all preparations implies equality of the corresponding response
functions,
\begin{equation}
    \xi_{g,k}(\lambda)=\xi_{g',k}(\lambda)
    \quad \forall \lambda,k.
\end{equation}
This shows that Spekkens's toy model and QSL are also Spekkens's noncontextual.

\section{Proof of Theorem \ref{theorem: classical system is univ NC}}
\label{app: proof of theorem NC for classical system}
We prove here that a classical system as we defined it, must be Spekkens noncontextual.

Consider a \textbf{Classical} system (a system which satisfies \textbf{Classical Assumptions~\ref{assumptclass1}, \ref{assumptclass2}, \& \ref{assumptclass3}}). Let $P_1$ and $P_2$ be two preparations that are operationally equivalent, i.e. that together respect Eq.~\eqref{eq: prep equivalence}. Then, for each measurable
$\omega \in \Sigma$, we can insert the definition of the response function from Eq.~\eqref{eq:responseindicatorclassical} into Eq.~\eqref{eq: prep equivalence} for measurement $M_\omega$ and outcome $k=1$ to get
\begin{equation}
    \int_\Lambda \mathbb I_\omega(\lambda)\,\mu_{P_1}(\lambda)d\lambda
    =
    \int_\Lambda \mathbb I_\omega(\lambda)\,\mu_{P_2}(\lambda)d\lambda
    \quad \forall\,\omega\in\Sigma.
\end{equation}
(and similar for $k=0$).
By the definition of integration of indicator functions~\cite{folland1999real}, this is
equivalent to
\begin{equation}
    \mu_{P_1}(\omega) = \mu_{P_2}(\omega)
    \quad \forall\,\omega\in\Sigma.
\end{equation}
Hence the two measures \(\mu_{P_1}\) and \(\mu_{P_2}\) coincide on all
measurable sets, and therefore
\begin{equation}
    \mu_{P_1}(\lambda) = \mu_{P_2}(\lambda)\ \text{Lebesgue a.e.}
\end{equation}

To be precise, the two functions are not perfectly equivalent, but equivalent almost everywhere. This means that equivalence holds for all elements in the set except a subset of measure zero. However, for our purposes ``equivalent almost everywhere'' can be treated as being the same as ``perfectly equivalent'', because we we always prepare our system in a macrostate (i.e. a uniform distribution over a measurable set), which is always a subset of nonzero measure.

We have therefore shown that any two operationally equivalent preparations must be represented by the same ontic probability measure. In other words,
\begin{equation}
    P_1 \sim P_2 \;\Rightarrow\; \mu_{P_1} = \mu_{P_2}
\end{equation}
which is means classical systems obey Spekkens' notion of preparation noncontextuality.

Let us now show that \textbf{Classical} systems must be Spekkens measurement noncontextual. Let $M_1$ and $M_2$ be two measurements that are operationally equivalent, i.e. that satisfy Eq.~\eqref{eq: meas equivalence}. By applying the definition of operational equivalence from Eq.~\eqref{eq: meas equivalence} to the specific class of preparations $P_\omega$ defined through Eq.~\eqref{eq: indicator functions for prep}, we obtain
\begin{equation}
    \int_\Lambda \xi_{M_1,k}(\lambda) \frac{\mathbb{I}_\omega(\lambda)}{V(\omega)} d\lambda
    =
    \int_\Lambda \xi_{M_2,k}(\lambda) \frac{\mathbb{I}_\omega(\lambda)}{V(\omega)} d\lambda,
\end{equation}
$\forall\,\omega \in \Sigma, \forall\,k$.

Since $V(\omega)$ is a constant for a given $\omega$, we can multiply both sides by $V(\omega)$. Using the property that the indicator function restricts the domain of integration, the equality simplifies to
\begin{equation}
    \int_\omega \xi_{M_1,k}(\lambda) d\lambda = \int_\omega \xi_{M_2,k}(\lambda) d\lambda \quad \forall\,\omega \in \Sigma, \forall\,k.
\end{equation}

By a fundamental theorem of measure theory~\cite{folland1999real}, if the integrals of two measurable functions are equal over every measurable set $\omega$, then the functions themselves must be equal almost everywhere. Thus:
\begin{equation}
    \xi_{M_1,k}(\lambda) = \xi_{M_2,k}(\lambda)
    \quad \text{Lebesgue a.e.}, \forall\,k.
\end{equation}
Here, the same argument that presented before for preparations is again valid. Since we assume our ability to perform a dichotomic measurement for a measurable set, which never has zero dimension, we can conclude that ``being equivalent almost everywhere'' can be treated the same as ``being perfectly equivalent''.

This demonstrates that for a classical system where macroscopic volumes of phase space can be prepared, any two operationally equivalent measurements must be represented by identical ontic response functions. This satisfies the definition of measurement noncontextuality
\begin{equation}
    M_1 \sim M_2 \;\Rightarrow\; \xi_{M_1} = \xi_{M_2}
\end{equation}

Finally, let us  show that a classical system is Spekkens transformation noncontextual. Assume the system allows for the same macroscopic preparations and measurements defined previously, i.e. Eqs.~\eqref{eq:responseindicatorclassical} and \eqref{eq: indicator functions for prep} hold.

Under these assumptions, let $T_1$  and $T_2$ be two equivalent transformations, i.e. that satisfy Eq.~\eqref{eq: transf equivalence}. Applying the operational equivalence from Eq.~\eqref{eq: transf equivalence} to a uniform preparation over $\omega$ and a measurement of the set $\omega'$, we obtain
\begin{equation}
\begin{split}
    \int_\Lambda d\lambda' \int_\Lambda d\lambda \, \mathbb{I}_{\omega'}(\lambda') \Gamma_{T_1}(\lambda'|\lambda) \frac{\mathbb{I}_{\omega}(\lambda)}{V(\omega)} = \int_\Lambda d\lambda' \int_\Lambda d\lambda \, \mathbb{I}_{\omega'}(\lambda') \Gamma_{T_2}(\lambda'|\lambda) \frac{\mathbb{I}_{\omega}(\lambda)}{V(\omega)}.
    \end{split}
\end{equation}

Multiplying by $V(\omega)$ and using the properties of indicator functions to restrict the domains of integration, this simplifies to
\begin{equation}
\begin{split}
    \int_{\omega} d\lambda \left( \int_{\omega'} \Gamma_{T_1}(\lambda'|\lambda) d\lambda' \right) = \int_{\omega} d\lambda \left( \int_{\omega'} \Gamma_{T_2}(\lambda'|\lambda) d\lambda' \right)
\end{split}
\end{equation}

Since this equality holds for every measurable set $\omega$, the inner integrals (which are functions of $\lambda$) must be equal almost everywhere:
\begin{equation}
    \int_{\omega'} \Gamma_{T_1}(\lambda'|\lambda) d\lambda' = \int_{\omega'} \Gamma_{T_2}(\lambda'|\lambda) d\lambda',\,\ \text{Lebesgue a.e.}, \;\forall \omega' \in \Sigma.
\end{equation}

Furthermore, since this must hold for every measurable set $\omega'$, the kernels themselves must coincide:
\begin{equation}
    \Gamma_{T_1}(\lambda'|\lambda) = \Gamma_{T_2}(\lambda'|\lambda) \quad \text{Lebesgue a.e.}
\end{equation}
The same arguments about nonperfect equivalence presented before can be applied here. Therefore, for \textbf{Classical} systems, two transformations being operationally equivalence implies those transformations are ontically identical:
\begin{equation}
    T_1 \sim T_2 \;\Rightarrow\; \Gamma_{T_1} = \Gamma_{T_2}
\end{equation}
which is the definition of Spekkens transformation noncontextuality.

This concludes our demonstration: a \textbf{Classical} system (i.e., one which satisfies \textbf{Classical Assumptions~\ref{assumptclass1}, \ref{assumptclass2}, \& \ref{assumptclass3}}) must be Spekkens noncontextual. The same is clearly true for a \textbf{Strongly Classical} system (where \textbf{Classical Assumptions~\ref{assumptclass2} \& \ref{assumptclass3}} are replaced by \textbf{\ref{assumptclass4} \& \ref{assumptclass5}} respectively).

\section{Gaussian Quantum Mechanics violates classical assumptions}
\label{app: Gaussian QM violates classical assumptions}
We prove here that Spekkens noncontextuality does not imply \textbf{Classicality} for a system. To prove this we show that Gaussian Quantum Mechanics (which is Spekkens noncontextual) violates our \textbf{Classical Assumptions}.

In Gaussian Quantum Mechanics (GQM) the ontic state space $\Lambda$ is the classical phase space $\mathbb{R}^{2n}$, so $\lambda = (\vec{r}, \vec{p})$. This means GQM satisfies \textbf{Classical Assumption~\ref{assumptclass1}}.
The preparations $P$ correspond to quantum states (density matrices $\rho$), represented in the ontological model by their Wigner functions $W_\rho(\lambda)$. For Gaussian states, by definition, $W_\rho(\lambda)$ is a non-negative Gaussian probability distribution over $\Lambda$.
Measurements and transformations also map to non-negative Gaussian functions and Gaussian transition kernels respectively.

Firstly we prove that GQM violates \textbf{Classical Assumption~\ref{assumptclass3}}.

\begin{proof}
In GQM, the set of valid preparation procedures $\{P\}$ is restricted strictly to Gaussian states. Therefore, the probability distribution $\mu_P(\lambda)$ associated with any valid preparation must be a Gaussian distribution of the form:
\begin{equation}
    \mu_P(\lambda) = \frac{1}{\sqrt{(2\pi)^{2n} \det(V)}} \exp\left(-\frac{1}{2}(\lambda - d)^T V^{-1} (\lambda - d)\right)
\end{equation}
where $d$ is the displacement vector (mean) and $V$ is the covariance matrix.

\textbf{Classical Assumption~\ref{assumptclass3}} requires that for every finite-measure set $\omega$, there exists a valid preparation with the distribution:
\begin{equation*}
    \mu_{P_\omega}(d\lambda) = \frac{\mathbb{I}_\omega(\lambda)}{V(\omega)} d\lambda
\end{equation*}
where $\mathbb{I}_\omega(\lambda)$ is the indicator function.

Let $\omega$ be a finite hypercube in phase space, e.g., $\omega = [q_0, q_1] \times [p_0, p_1]$. The indicator function $\mathbb{I}_\omega(\lambda)$ is a step function  and it has bounded, compact support. A step function with compact support cannot be written in the form of a Gaussian distribution. Furthermore, a Gaussian distribution has infinite support (it is strictly positive everywhere in $\mathbb{R}^{2n}$, approaching zero only as $|\lambda| \to \infty$), whereas $\mathbb{I}_\omega(\lambda) = 0$ for all $\lambda \notin \omega$.
Therefore, the uniform preparation $P_\omega$ required by \textbf{Classical Assumption~\ref{assumptclass3}} is not a valid preparation in Gaussian Quantum Mechanics. GQM violates \textbf{Classical Assumption~\ref{assumptclass3}}.
\end{proof}

Moreover GQM violates also \textbf{Classical Assumption~\ref{assumptclass2}}.
\begin{proof}
\textbf{Classical Assumption~\ref{assumptclass2}} requires that for \textit{every} measurable set $\omega \in \Sigma$, there exists a measurement with a response function $\xi(\lambda) = \mathbb{I}_\omega(\lambda)$. This represents an infinitely sharp measurement of whether the ontic state $\lambda$ is perfectly inside the arbitrary boundary of $\omega$.

In GQM, valid measurements correspond to Gaussian POVMs. The response functions $\xi_k(\lambda)$ in the ontological model are given by the Wigner representation of the POVM elements $E_k$.

For a POVM element to be Gaussian, its Wigner function must have a Gaussian profile. The indicator function $\mathbb{I}_\omega(\lambda)$ of a region $\omega$ (like our hypercube) is discontinuous at the boundaries of $\omega$. A Gaussian function is $C^\infty$ (smooth and infinitely differentiable everywhere). 

Therefore, there exists no Gaussian POVM element $E$ whose Wigner representation yields the sharp step-function $\mathbb{I}_\omega(\lambda)$. GQM restricts measurements to those with Gaussian, smooth response functions, thereby prohibiting the sharp, dichotomous indicator response functions required by \textbf{Classical Assumption~\ref{assumptclass2}}.
\end{proof}

\section{Derivation of Bell CHSH inequality}
\label{appendix: derivation of CHSH inequality}

We derive here Bell inequality in the CHSH form.

We consider a generic model which allows us to specify, in the most exhaustive way possible, the state of the system, where such a specification uniquely determines the probabilities of different measurement results. 
In our model, we consider two particles whose spin is measured along arbitrary directions, respectively $\vec{a}$ and $\vec{b}$, and whose spin components can only take the values $+1$ or $-1$, representing the possible outcomes of the spin measurement. 
Let $\lambda$ denote the set of all variables (available or hidden) whose specification  in the most exhaustive way possible allows the theory under consideration to determine  the state of the system. 

We denote by $p^{AB}_{\lambda}(\vec{a},\vec{b};\alpha,\beta)$ the probability of obtaining the outcomes  in the measurement of the spin components of the two particles along the directions 
$\vec{a}$ and $\vec{b}$ in the opposite regions of detectors $A$ and $B$,  in a system characterised by the variables $\lambda$. 
The only hypothesis made by Bell is that of locality, called Bell locality,  namely that the probability of obtaining two measurement outcomes at the two extremes of the apparatus is given by the product of the two separate probabilities (as shown in Eq.~\eqref{eq: Bell locality assumption}).

Now consider a function $E_{\lambda}(\vec{a},\vec{b})$ defined as the sum of the probabilities of obtaining concordant outcomes minus the sum of the probabilities  of obtaining discordant outcomes:
\begin{equation}
E_{\lambda}(\vec{a},\vec{b}) =
p^{AB}_{\lambda}(\vec{a},\vec{b};+1,+1)
+
p^{AB}_{\lambda}(\vec{a},\vec{b};-1,-1)
-
p^{AB}_{\lambda}(\vec{a},\vec{b};+1,-1)
-
p^{AB}_{\lambda}(\vec{a},\vec{b};-1,+1)
\end{equation}
It is possible to write $E_{\lambda}(\vec{a},\vec{b})$ using Bell locality:
\begin{equation}
E_{\lambda}(\vec{a},\vec{b}) =
\big[p^{A}_{\lambda}(\vec{a};*,+1) - p^{A}_{\lambda}(\vec{a};*,-1)\big]\cdot
\big[p^{B}_{\lambda}(\vec{b};*,+1) - p^{B}_{\lambda}(\vec{b};*,-1)\big]
\end{equation}
Considering now a third arbitrary direction $\vec{d}$ one can write:
\begin{equation}
\label{eq: Bell factor def}
E_{\lambda}(\vec{a},\vec{b}) - E_{\lambda}(\vec{a},\vec{d})
=
\big[p^{A}_{\lambda}(\vec{a};*,+1) - p^{A}_{\lambda}(\vec{a};*,-1)\big]\cdot
\big[
p^{B}_{\lambda}(\vec{b};*,+1) - p^{B}_{\lambda}(\vec{b};*,-1)
-
p^{B}_{\lambda}(\vec{d};*,+1) + p^{B}_{\lambda}(\vec{d};*,-1)
\big]
\end{equation}
Since the measurement under examination has only two possible outcomes, it holds that:
\begin{equation}
p^{A}_{\lambda}(\vec{a};*,+1) + p^{A}_{\lambda}(\vec{a};*,-1) = 1
\end{equation}
From this the first factor of Eq.~\eqref{eq: Bell factor def} becomes:

\begin{equation}
p^{A}_{\lambda}(\vec{a};*,+1) - p^{A}_{\lambda}(\vec{a};*,-1)
= 1 - 2\,p^{A}_{\lambda}(\vec{a};*,-1)
\end{equation}
Since $p^{A}_{\lambda}(\vec{a};*,-1)$ is a probability, it must have a value between $0$ and $1$, from which it follows that the value of the expression on the right-hand  side of the previous equation is between $-1$ and $1$. Consequently its absolute value will be less than $1$. Hence:

\begin{equation}
\left|E_{\lambda}(\vec{a},\vec{b}) - E_{\lambda}(\vec{a},\vec{d})\right|
\le
\left|
\big[p^{B}_{\lambda}(\vec{b};*,+1) - p^{B}_{\lambda}(\vec{b};*,-1)\big]- 
\big[p^{B}_{\lambda}(\vec{d};*,+1) - p^{B}_{\lambda}(\vec{d};*,-1)\big]
\right|
\end{equation}
The same reasoning can be followed to obtain:
\begin{equation}
\left|E_{\lambda}(\vec{c},\vec{b}) + E_{\lambda}(\vec{c},\vec{d})\right|
\le
\left|
\big[p^{B}_{\lambda}(\vec{b};*,+1) - p^{B}_{\lambda}(\vec{b};*,-1)\big]
+
\big[p^{B}_{\lambda}(\vec{d};*,+1) - p^{B}_{\lambda}(\vec{d};*,-1)\big]
\right|
\end{equation}
Summing the last two equations one obtains:
\begin{equation}
\left|E_{\lambda}(\vec{a},\vec{b}) - E_{\lambda}(\vec{a},\vec{d})\right|
+
\left|E_{\lambda}(\vec{c},\vec{b}) + E_{\lambda}(\vec{c},\vec{d})\right|
\le
|r-s| + |r+s|
\end{equation}
where we set
\begin{equation}
    r = p^{B}_{\lambda}(\vec{b};*,+1) - p^{B}_{\lambda}(\vec{b};*,-1),\;\;\;\;
    s = p^{B}_{\lambda}(\vec{d};*,+1) - p^{B}_{\lambda}(\vec{d};*,-1).
\end{equation}

The expression on the right-hand side of the inequality has maximum value $2$. 
From this it follows that:
\begin{equation}
\left|E_{\lambda}(\vec{a},\vec{b}) - E_{\lambda}(\vec{a},\vec{d})\right|
+
\left|E_{\lambda}(\vec{c},\vec{b}) + E_{\lambda}(\vec{c},\vec{d})\right|
\le 2
\end{equation}

The final step consists in considering not the functions $E_{\lambda}$ but their averages, defined as:

\begin{equation}
E(\vec{m},\vec{n}) =
\int E_{\lambda}(\vec{m},\vec{n})\,p(\lambda)\,d\lambda
\end{equation}

where $p(\lambda)$ is the distribution of the variables that characterise the system. One then obtains:
\begin{equation}
\left|E(\vec{a},\vec{b}) - E(\vec{a},\vec{d})\right|
+
\left|E(\vec{c},\vec{b}) + E(\vec{c},\vec{d})\right|
\le 2
\end{equation}
This form of Bell's inequality proposed here is the one obtained in \cite{Clauser1969Proposed}, called CHSH (from authors' initials), which is equivalent to Bell's original inequality.
This version of Bell's inequality is the one mainly used in literature because it provides a generalisation with respect to the original one \cite{Bell1966}. The original derivation indeed considered only perfectly anti-correlated measurements while the CHSH version takes into account also measurements which are non-perfectly anti-correlated and it is then more suitable for studying nonlocality both theoretically and experimentally.

\section{Proof of Prop. \ref{prop: Sheaf, global section = NC}}
\label{app: Proof of NC = global section}
We prove here that an empirical model can be described by a factorisable hidden-variable model (i.e. NCHV model) if and only if a global section of that model exists.
\begin{proof}
\textbf{(1 $\implies$ 2):} Assume there exists a global section $d \in D_R\mathcal{E}(X)$. We construct a factorisable hidden-variable model as follows:
\begin{enumerate}
    \item Let the set of hidden variables $\Lambda$ be the set of all global assignments $\mathcal{E}(X)$.
    \item Let the distribution over hidden variables $h_\Lambda \in D_R(\Lambda)$ be the global section itself: $h_\Lambda(\lambda) = d(\lambda)$.
    \item For each $\lambda \in \Lambda$ and context $C \in \mathcal{M}$, define the local distribution $h_C^\lambda \in D_R\mathcal{E}(C)$ as the deterministic Dirac distribution:
    \begin{equation}
        h_C^\lambda(s) = \delta_{\lambda|_C}(s) = 
        \begin{cases} 
        1 & \text{if } \lambda|_C = s \\
        0 & \text{otherwise}
        \end{cases}
    \end{equation}
\end{enumerate}

To show factorisability, note that for any $s \in \mathcal{E}(C)$:
\begin{equation}
    h_C^\lambda(s) = \delta_{\lambda|_C}(s) = \prod_{x \in C} \delta_{\lambda|_{\{x\}}}(s|_{\{x\}})
    = \prod_{x \in C} h_{\{x\}}^\lambda(s|_{\{x\}})
\end{equation}
To show realisation, we average over $\Lambda$:
\begin{equation}
    \sum_{\lambda \in \Lambda} h_C^\lambda(s) h_\Lambda(\lambda) = \sum_{\lambda \in \mathcal{E}(X)} \delta_{\lambda|_C}(s) d(\lambda)
    = \sum_{\substack{\lambda \in \mathcal{E}(X) \lambda|_C = s}} d(\lambda) = d|_C(s) = e_C(s)
\end{equation}
Thus, the global section induces a factorisable HV model.

\vspace{1em}
\textbf{(2 $\implies$ 1):} Assume $e$ is realised by a factorisable HV model $(\Lambda, h_\Lambda, \{h_C^\lambda\})$. By factorisability, for each $\lambda$, the distribution $h_C^\lambda$ is determined by its marginals on singletons: $h_C^\lambda(s) = \prod_{x \in C} h_{\{x\}}^\lambda(s|_{\{x\}})$.

We define a global distribution $d \in D_R\mathcal{E}(X)$ by:
\begin{equation}
    d(g) = \sum_{\lambda \in \Lambda} \left( \prod_{x \in X} h_{\{x\}}^\lambda(g|_{\{x\}}) \right) h_\Lambda(\lambda) \quad \forall g \in \mathcal{E}(X)
\end{equation}
We verify that $d$ restricts to $e_C$ for any $C \in \mathcal{M}$. For $s \in \mathcal{E}(C)$:
\begin{equation}
\begin{split}
    d|_C(s) = \sum_{\substack{g \in \mathcal{E}(X) \\ g|_C = s}} d(g)
    = \sum_{\lambda \in \Lambda} h_\Lambda(\lambda) \sum_{\substack{g \in \mathcal{E}(X) \\ g|_C = s}} \left( \prod_{x \in X} h_{\{x\}}^\lambda(g|_{\{x\}}) \right)
\end{split}
\end{equation}
Splitting the product into $x \in C$ and $x \in X \setminus C$:
\begin{equation}
\begin{aligned}
    d|_C(s) = \sum_{\lambda \in \Lambda} h_\Lambda(\lambda) \left( \prod_{x \in C} h_{\{x\}}^\lambda(s|_{\{x\}})\right) \left[ \sum_{g' \in \mathcal{E}(X \setminus C)} \prod_{y \in X \setminus C} h_{\{y\}}^\lambda(g'|_{\{y\}}) \right]
\end{aligned}
\end{equation}
The term in the square brackets is a sum over all possible outcomes for measurements in $X \setminus C$, which marginalises to 1. Using the factorisability of the HV model:
\begin{equation}
    d|_C(s) = \sum_{\lambda \in \Lambda} h_\Lambda(\lambda) h_C^\lambda(s) = e_C(s)
\end{equation}
Thus, $d$ is a global section for the empirical model $e$.
\end{proof}

\section{Proof of Prop. \ref{prop: state dep context implies Bell NL}}
\label{Appendix: Cabello state dep context implies Bell NL}
We prove here that any set of measurements exhibiting State-Dependent Contextuality can be mapped to a Bell inequality that is violated by a maximally entangled state.

To prove this we need the following result \cite{Cabello1996}. In $d\geq3$, given any two non-orthogonal rank-one projectors $\Pi_A$ and $\Pi_B$, there is a set of projectors $E$ such that, for any KS assignment $f$, $f(\Pi_A)+f(\Pi_B)\leq 1$. The set $\Pi_A\cup E\cup \Pi_B$ is called a true-implies-false set (TIFS) \cite{Cabello2018TIFS}. \\
Since any quantum contextual behaviour can be produced by a set of rank-one projectors we consider the set of rank one projectors $S=\{\Pi_1,...,\Pi_n\}$. In this set the contexts are the subsets containing mutually commuting projectors.
To construct the critical SI-C set we consider the graph $G$ of orthogonality of $S$, i.e. the graph which connects with the same edge vertices corresponding to mutually orthogonal projectors. Let $N$ be the minimum number of disjoint bases that cover all the vertices of $G$. Since $S$ allows for SD-C, then $N\geq 3$ \cite{CabelloSeveriniWinter2014}. If $N<d+1$, then we add disjoint bases until the total of number of disjoint bases is $N+1$.
\begin{figure}[H]
    \centering
    \includegraphics[width=0.5\linewidth]{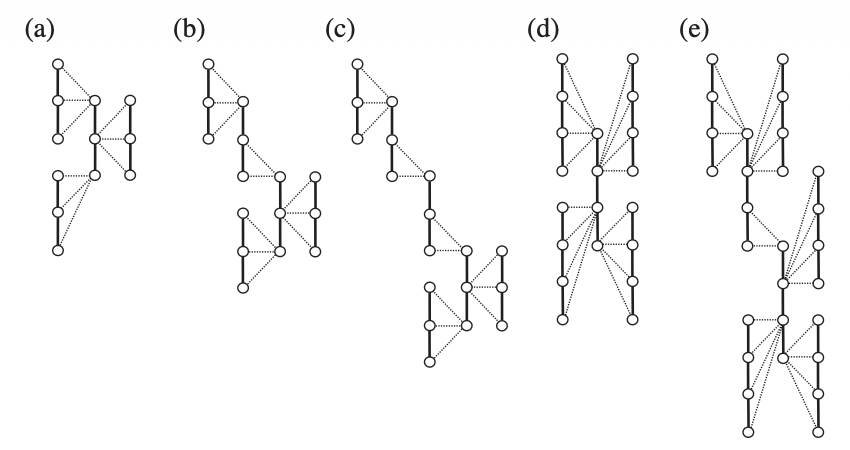}
    \caption{Construction of SD-C set for $d=3$ and $d=4$, image taken from \cite{Cabello2021ConvertingCtoNL}}
    \label{fig:CabelloSDC}
\end{figure}
Fig.~\ref{fig:CabelloSDC} shows the process of constructing the critical SI-C set. Here, every node represents a rank-1 projector. A continuous vertical line between $d\geq 3$ nodes indicates that they are mutually orthogonal. Hence, in dimension $d$, in any KS assignment, one of them has to be assigned 1. A dashed line between two nodes indicates that there is a TIFS between (and including) them. Hence, in any KS assignment, both of them cannot be assigned 1. (a) shows the case of $d=3$ and $N=d+1$, (b) $d=3$ and $N=d+2$, (c) $d=3$ and $N=d+3$, (d) $d=4$ and $N=d+1$, (e) $d=4$ and $N=d+2$. The construction works similarly for any $d \geq 3$ and $N\geq d+1$. In all cases, it is impossible to assign to the depicted nodes the values 0 or 1 satisfying that one of the $d$ nodes in each continuous vertical line must be 1, while nodes connected by a dashed line cannot both be 1. However, such an assignment is possible whenever we remove any of the depicted nodes. In all cases, the resulting set is a critical KS set in dimension $d$. However this does not imply that the set is a SI-C set. This however can be solved suitably choosing the extra nodes used for the TIFSs \cite{Budroni2022KSContextuality}.
This construction allows us to create a SI-C set from a SD-C set.
In \cite{Cabello2015NecessaryAndSufficient} it is proved that a set of projectors $S''=\{\Pi_1,...,\Pi_n\}$ is a SI-C set if and only if there are non-negative numbers $w=(w_1,...,w_n)$ and a number $0\leq y<1 $ such that $\sum_{j\in \mathcal{I}}w_j\leq y$ for all $\mathcal{I}$ where $\mathcal{I}$ is the independent set of the graph $G$ and $\sum_iw_i\Pi_i\geq I$. An independent set of a graph is a set of vertices in a graph, no two of which are adjacent.
Based on this condition we can write the following inequality, which is valid for a noncontextual hidden variable model~\cite{CabelloSeveriniWinter2014}:
\begin{equation}
\label{eq: NC inequality CabelloSevWint}
\sum_{i \in V(G)} w_i P(\Pi_i = 1) - \sum_{(i,j) \in E(G)} \max(w_i, w_j)\cdot p(\Pi_i = 1, \Pi_j = 1) \myeq \alpha(G,w),
\end{equation}
where $p(\Pi_i = 1, \Pi_j = 1)$ is the probability of obtaining outcome 1 in the measurement associated to $\Pi_i$ and also in the measurement associated to $\Pi_j$, $V(G)$ and $E(G)$ are the sets of vertices and edges of $G$, respectively, $\alpha(G,w)$ is the independence number of $(G,w)$ [i.e., the graph in which weight $w_i$ is assigned to each $i \in V(G)$]. The independence number of a (weighted) graph is the cardinality of its largest set of vertices (taking their weights into account) such that no two are adjacent.

The next step is the following: in each run of the experiment, we prepare a pair of particles in the two-qudit maximally entangled state:
\begin{equation}
\label{eq: max entangled qudit}
|\Psi\rangle = \frac{1}{\sqrt{d}} \sum_{k=0}^{d-1} |kk\rangle,
\end{equation}
distribute one particle to Alice and the other to Bob, and allow Alice (Bob) to freely and independently choose and perform one measurement from $\mathcal{S}''$ (from the set obtained by taking the complex conjugate of the elements in $\mathcal{S}''$).
The behaviour produced by this state and these measurements violate the following Bell inequality:
\begin{equation}
\begin{aligned}
\label{eq: Bell inequality from SIC}
\sum_{i \in V(G)} w_i P(\Pi_i^A = 1, \Pi_i^B = 1)
- \sum_{(i,j) \in E(G)} \frac{\max(w_i, w_j)}{2}\cdot
\left[
P(\Pi_i^A = 1, \Pi_j^B = 1) + P(\Pi_j^A = 1, \Pi_i^B = 1)
\right]
\myeqq \alpha(G,w),
\end{aligned}
\end{equation}
where $P(\Pi_i^A = 1, \Pi_j^B = 1)$ is the probability that Alice obtains outcome 1 for measurement $\Pi_i$ on her particle and Bob obtains outcome 1 for measurement $\Pi_j$ on his particle. 

That Eq.~\eqref{eq: Bell inequality from SIC} is a Bell inequality follows from the fact that, for Local Hidden Variable Theories, the maximum of the left-hand side of Eq.~\eqref{eq: Bell inequality from SIC}  is always attained by a deterministic assignment for the outcomes of the elements of $\mathcal{S}''$ in Alice's particle and a deterministic assignment for the outcomes of the elements of the complex conjugate of $\mathcal{S}''$ in Bob's particle. To maximise the left-hand side of Eq.~\eqref{eq: Bell inequality from SIC}, we need to maximise (taking into account the weights) the number of projectors $\Pi_i$ to which outcome 1 is assigned both in Alice's and Bob's particles, while minimising the number of adjacent $\Pi_j$ to which outcome 1 is assigned, which is exactly the definition of independence number of a (weighted) graph $(G,w)$. The quantum violation of Eq.~\eqref{eq: NC inequality CabelloSevWint} for the maximally mixed state using S'' is equal to the quantum violation of the Bell inequality \ref{eq: Bell inequality from SIC} for the
maximally entangled state \ref{eq: max entangled qudit} and using $\mathcal{S}''$ in Alice’s side and the complex conjugate of $\mathcal{S}''$ in Bob’s side.

With this argument we demonstrated that noncontextuality inequalities of the form of Eq.~\eqref{eq: NC inequality CabelloSevWint} are in one-to-one correspondence with Bell inequalities of the form of Eq.~\eqref{eq: Bell inequality from SIC}. This concludes the demonstration of Prop.~\ref{prop: state dep context implies Bell NL}.

\section{Proof of Theorem \texorpdfstring{\ref{theorem: mixed qubit implies prep contextuality}}{6}}
\label{appendix: proof of theorem mixed qubit prep context}
We prove here that every mixed state of a qubit is preparation contextual.

A qubit in any mixed (i.e., non-pure) state can be represented by the density matrix $\rho_n=\frac{1}{2}(\mathds{1}+r_x\sigma_x+r_y\sigma_y+r_z\sigma_z)$ where $\mathds{1}$ is the identity matrix, $\vec{r}=(r_x,r_y,r_z)$ is the Bloch vector with $0\leq|\vec{r}|<1$ and $\sigma_{x,y,z}$ are Pauli matrices. 
This state can have multiple different decompositions, where we can associate each decomposition to a different preparation procedure. Consider the following 6 decompositions:
\begin{equation}
\begin{aligned}
\label{eq:6 decompositions of rho}
\rho_n =& \frac{1 - q}{2} \ket{\phi_n^\perp}\bra{\phi_n^\perp} + \frac{1 + q}{2} \ket{\phi_n}\bra{\phi_n} 
= \frac{1 - q}{2} \left( \ket{\psi_a}\bra{\psi_a} + \ket{\psi_a^\perp}\bra{\psi_a^\perp} \right) + q \ket{\phi_n}\bra{\phi_n}  \\
=& \frac{1 - q}{2} \left( \ket{\psi_b}\bra{\psi_b} + \ket{\psi_b^\perp}\bra{\psi_b^\perp} \right) + q \ket{\phi_n}\bra{\phi_n}  
= \frac{1 - q}{2} \left( \ket{\psi_c}\bra{\psi_c} + \ket{\psi_c^\perp}\bra{\psi_c^\perp} \right) + q \ket{\phi_n}\bra{\phi_n}  \\
=& \frac{1 - q}{3} \left( \ket{\psi_a}\bra{\psi_a} + \ket{\psi_b}\bra{\psi_b} + \ket{\psi_c}\bra{\psi_c} \right)+ q\ket{\phi_n}\bra{\phi_n}\\  
=& \frac{1 - q}{3} \left( \ket{\psi_a^\perp}\bra{\psi_a^\perp} + \ket{\psi_b^\perp}\bra{\psi_b^\perp} + \ket{\psi_c^\perp}\bra{\psi_c^\perp} \right) + q \ket{\phi_n}\bra{\phi_n} 
\end{aligned}
\end{equation}

\noindent
where $\ket{\phi_n}\bra{\phi_n} = \frac{1}{2}(\mathds{1}+\hat{n}\cdot\sigma)$ where $\hat{n}=\vec{r}/|\vec{r}|$ and the vectors 
$\ket{\psi_a}, \ket{\psi_b}, \ket{\psi_c}$ are chosen from the equatorial plane of the Bloch sphere
perpendicular to $\hat{n}$ such that the (diameter) line joining $\ket{\psi_a}$ and $\ket{\psi_a^\perp}$ makes a $\pi/3$ angle with both the (diameter) line joining $\ket{\psi_b}$ and $\ket{\psi_b^\perp}$, and the (diameter) line joining $\ket{\psi_c}$ and $\ket{\psi_c^\perp}$.
Whenever two density operators are orthogonal, the associated preparation procedures can be distinguished with certainty in a single shot measurement. Further, whenever two preparation procedures are distinguishable with certainty in a single shot measurement, their associated probability distribution must be non-overlapping \cite{Spekkens2005}. 

Due to this,
\begin{equation}
\label{eq: orthogonal ontic ditributions}
\begin{aligned}
\mu(\lambda|\phi_n)\mu(\lambda|\phi_n^\perp) &= 0,\;\;\;\;
\mu(\lambda|\psi_a)\mu(\lambda|\psi_a^\perp)= 0,\;\;\;\;
\mu(\lambda|\psi_b)\mu(\lambda|\psi_b^\perp) = 0,\;\;\;\;
\mu(\lambda|\psi_c)\mu(\lambda|\psi_c^\perp) = 0. 
\end{aligned}
\end{equation}
We can associate the six decompositions in Eq.~\eqref{eq:6 decompositions of rho} of $\rho_n$ with six different preparation procedures: 
$C_{\phi_n^\perp \phi_n}, C_{\psi_a \psi_a^\perp \phi_n}, \ldots, C_{\psi_b^\perp \psi_c^\perp \phi_n}$ respectively.
In an ontological model, a convex combination of preparation procedures is represented by a convex sum of associated probability distributions \cite{Spekkens2005}, meaning 
\begin{equation}
\label{eq: distributions convex sums}
    \begin{aligned}
&\mu(\lambda|\rho_n, \mathcal{C}_{\phi_n^\perp \phi_n}) = \frac{1-q}{2}\mu(\lambda|\phi_n^\perp) + \frac{1+q}{2}\mu(\lambda|\phi_n)  \\&
\mu(\lambda|\rho_n, \mathcal{C}_{\psi_a \psi_a^\perp \phi_n}) = \frac{1-q}{2}[\mu(\lambda|\psi_a) + \mu(\lambda|\psi_a^\perp)] + q\mu(\lambda|\phi_n)  \\&
\mu(\lambda|\rho_n, \mathcal{C}_{\psi_b \psi_b^\perp \phi_n}) = \frac{1-q}{2}[\mu(\lambda|\psi_b) + \mu(\lambda|\psi_b^\perp)] + q\mu(\lambda|\phi_n) \\&
\mu(\lambda|\rho_n, \mathcal{C}_{\psi_c \psi_c^\perp \phi_n}) = \frac{1-q}{2}[\mu(\lambda|\psi_c) + \mu(\lambda|\psi_c^\perp)] + q\mu(\lambda|\phi_n) \\&
\mu(\lambda|\rho_n, \mathcal{C}_{\psi_a \psi_b \psi_c \phi_n}) = \frac{1-q}{3}[\mu(\lambda|\psi_a) + \mu(\lambda|\psi_b) + \mu(\lambda|\psi_c)]+ q\mu(\lambda|\phi_n) \\&
\mu(\lambda|\rho_n, \mathcal{C}_{\psi_a^\perp \psi_b^\perp \psi_c^\perp \phi_n}) = \frac{1-q}{3}[\mu(\lambda|\psi_a^\perp) + \mu(\lambda|\psi_b^\perp) + \mu(\lambda|\psi_c^\perp)]+ q\mu(\lambda|\phi_n)
\end{aligned}
\end{equation}

Let us denote the support of $\mu(\lambda|\rho_n)$ by $\Lambda_{\rho_n}$, i.e., $\Lambda_{\rho_n} = \{\lambda \in \Lambda \mid \mu(\lambda|\rho_n) > 0\}$. As mentioned before, preparation noncontextuality requires the distribution over ontic states $\lambda$ associated with a preparation procedure to depend only upon the density matrix $\rho_n$ (i.e., that all preparations which give the same density matrix should give the same distribution $\mu(\lambda|\rho_n)$). This means if preparation noncontextuality is satisfied, the following equations should also hold:
\begin{equation}
\begin{aligned}
\label{eq: 6 ontic distributions}
\mu(\lambda|\rho_n) =& \frac{1-q}{2}\mu(\lambda|\phi_n^\perp) + \frac{1+q}{2}\mu(\lambda|\phi_n)  
=\frac{1-q}{2}[\mu(\lambda|\psi_a) + \mu(\lambda|\psi_a^\perp)] + q\mu(\lambda|\phi_n)  \\
=& \frac{1-q}{2}[\mu(\lambda|\psi_b) + \mu(\lambda|\psi_b^\perp)] + q\mu(\lambda|\phi_n) 
= \frac{1-q}{2}[\mu(\lambda|\psi_c) + \mu(\lambda|\psi_c^\perp)] + q\mu(\lambda|\phi_n) \\
=& \frac{1-q}{3}[\mu(\lambda|\psi_a) + \mu(\lambda|\psi_b) + \mu(\lambda|\psi_c)] + q\mu(\lambda|\phi_n)  
= \frac{1-q}{3}[\mu(\lambda|\psi_a^\perp) + \mu(\lambda|\psi_b^\perp) + \mu(\lambda|\psi_c^\perp)] + q\mu(\lambda|\phi_n) 
\end{aligned}
\end{equation}

However, there is no distribution which is compatible with both Eq.~\eqref{eq:6 decompositions of rho} and Eq.~\eqref{eq: 6 ontic distributions}. Indeed to satisfy Eq.~\eqref{eq: orthogonal ontic ditributions}, at any given $\lambda$ either $\mu(\lambda|\phi_n)$ or $\mu(\lambda|\phi_n^\perp)$ must be zero.
The same is true for the pairs $\{\mu(\lambda|\psi_a), \mu(\lambda|\psi_a^\perp)\}$, $\{\mu(\lambda|\psi_b), \mu(\lambda|\psi_b^\perp)\}$ and $\{\mu(\lambda|\psi_c), \mu(\lambda|\psi_c^\perp)\}$. Thus, we have sixteen different situations in total, but all of them lead to a contradiction, as shown in Ref.~\cite{Banik2014}. The above argument holds for any $\lambda\in \Lambda_{\rho_n}$. Therefore, a preparation-noncontextual assignment of ontic state $\lambda$ for the non-pure qubit state $\rho_n$ is not possible.

\end{appendices}
\end{document}